\documentclass[12pt,reqno]{amsart}

\usepackage[T1]{fontenc}
\usepackage[utf8]{inputenc}
\usepackage[english]{babel}
\usepackage{float}  
\usepackage{mathrsfs, csquotes, amsfonts,amssymb,
	amsmath,mathtools,amssymb, enumerate, xcolor, dsfont, mathabx}

\usepackage{graphicx} 
\usepackage[font=small,labelfont=bf]{caption} 
\usepackage[subrefformat=parens]{subcaption}

\theoremstyle{plain}
\newtheorem{theorem}{Theorem}[section]
\newtheorem{lemma}[theorem]{Lemma}

\newtheorem{proposition}[theorem]{Proposition}

\theoremstyle{definition}
\newtheorem{remark}[theorem]{Remark}

\newtheorem{definition}[theorem]{Definition}

\usepackage{color}
\definecolor{bleu_sombre}{rgb}{0,0,0.6}
\definecolor{Bl}{rgb}{0,0,0.6}
\definecolor{rouge_sombre}{rgb}{0.8,0,0}
\definecolor{vert_sombre}{rgb}{0,0.6,0}
\usepackage[plainpages=false,colorlinks,linkcolor=bleu_sombre,
citecolor=rouge_sombre,urlcolor=vert_sombre,breaklinks]{hyperref}

\usepackage{pgf,tikz,pgfplots}
\definecolor{webblue}{rgb}{0.22,0.45,0.70}
\definecolor{webred}{rgb}{0.5, 0.09, 0.09}
\definecolor{zzttqq}{rgb}{0.6,0.2,0.}
\usetikzlibrary{arrows}

\renewcommand{\leq}{\leqslant}	
\renewcommand{\geq}{\geqslant}
\newcommand{\C}{\mathbb{C}}
\newcommand{\R}{\mathbb{R}}
\newcommand{\N}{\mathbb{N}}

\newcommand{\dd}{\mathrm{d}}

\newcommand{\ad}{{\,\dagger}}

\makeatletter

\@addtoreset{equation}{section}
\makeatother

\numberwithin{equation}{section}

\title[]{Magnetic Dirac systems: Violation of  bulk-edge correspondence in the zigzag limit}
\author[J.-M. Barbaroux]{J.-M. Barbaroux}
\address[J.-M. Barbaroux]{Aix Marseille Univ, Universit\'e de Toulon, CNRS, CPT, Marseille, France}
\email{barbarou@univ-tln.fr}

\author{H. D. Cornean}
\address[H. D. Cornean]{Department of Mathematical Sciences, Aalborg University, Skjernvej 4A, 9220 Aalborg, Denmark}
\email{cornean@math.aau.dk}

\author[L. Le Treust]{L. Le Treust}
\address[L. Le Treust]{Aix Marseille Univ, CNRS, Centrale Marseille, I2M, Marseille, France}
\email{loic.le-treust@univ-amu.fr}

\author[N. Raymond]{N. Raymond}
\address[N. Raymond]{Univ Angers, CNRS, LAREMA, Institut Universitaire de France, SFR MATHSTIC, F-49000 Angers, France}
\email{nicolas.raymond@univ-angers.fr }

\author[E. Stockmeyer]{E. Stockmeyer}
\address[E. Stockmeyer]{Instituto de F\'isica, Pontificia Universidad Cat\'olica de Chile, Vicu\~na Mackenna 4860, Santiago 7820436, Chile.}
\email{stock@fis.puc.cl}
\begin{document}
	\maketitle
	\begin{abstract}
		We consider  a Dirac operator with  constant magnetic field defined on a half-plane with boundary conditions that interpolate between infinite mass and zigzag.  By a detailed study of the energy dispersion curves we show that the infinite mass case generically captures the profile of these curves, which undergoes a continuous pointwise  deformation into the topologically different zigzag profile. Moreover, these results are applied to the bulk-edge correspondence. In particular, by means of a counterexample, we show that this correspondence does not always hold true in the zigzag case. 		
	\end{abstract}

	\tableofcontents
	
	\section{Introduction and main results}

	Bulk-edge correspondence is a notion  that arises in the study of certain condensed matter physics systems possessing  non-trivial topology. It establishes a connection between the bulk properties of a  material (its interior or bulk region) and its edge or boundary properties. This correspondence may be given through an equation that links  a Chern number, that depends only on the bulk operator,  and  an expression involving the edge-states localized close to the boundary; eventually yielding the so-called topological quantization of  edge currents \cite{KELLENDONK2004388}.  Due to the integer (or topological) nature of the Chern number these relations are very stable under  smooth changes of the material parameters and therefore its importance for  potential applications.
	In this context, Dirac Hamiltonians are prominent examples  that exhibit interesting phenomena. They are used to  model various  materials, among them,    graphene and topological insulators \cite{Mong-Phys}.

	In this article we investigate the bulk-edge correspondence for  a  two-dimensional Dirac Hamiltonian with a constant magnetic field defined on a half-plane. This model was recently considered in  
	\cite{Cornean_2023}, where bulk-edge correspondence was shown  to hold, provided infinite-mass boundary conditions along  the edge are imposed. 
	Our motivation is to investigate  the validity of these results  when  any fixed admissible  local boundary condition is allowed. To this end we define a  family of Dirac Hamiltonians 
	$\mathscr{D}_{\gamma}$, where $\gamma\in \R\cup\{+\infty\}$ characterizes  the boundary conditions, in particular,   $\gamma=\pm 1$ corresponds to infinite mass.

	Our  main result Theorem~\ref{prop.bec} indicates that bulk-edge correspondence for this model still holds, provided the boundary conditions are not zigzag (\emph{i.e.} $\gamma\notin \{0,+\infty\}$). Moreover, Theorem~\ref{prop.bec} shows that this correspondence is violated for certain energies   when zigzag boundary conditions are imposed. 
	In order to show Theorem~\ref{prop.bec} certain knowledge on the energy dispersion curves associated to 
	$\mathscr{D}_{\gamma}$ is helpful.  In this work we present a fairly detailed analysis of them  extending the results of \cite{barbaroux:hal-02889558} to any local boundary condition. We complement our analysis with numerical illustrations of the energy dispersion curves for different values of the boundary parameter and  the magnetic field.

	Let us now turn to define our  model.  We consider a magnetic Dirac system on the half-plane denoted by $\mathbb{R}^2_+=\{(x_1,x_2)\in\mathbb{R}^2 : x_2>0\}$ in the presence of an orthogonal magnetic field whose component in the $x_3$ direction is given by 
	$b\in \R\setminus\{0\}$.
	The corresponding Hamiltonian  acts on functions in $L^2(\R^2_+,\mathbb{C}^2)$ as
	\begin{align}\label{e1}
		\sigma\cdot(-i\nabla-\mathbf{A})=\begin{pmatrix}
			0&-i\partial_1-\partial_2+bx_2\\
			-i\partial_1+\partial_2+bx_2	&0
		\end{pmatrix}\,.
	\end{align}
	Here $\mathbf{A}$ refers to a vector potential associated with the magnetic field \emph{i.e.} $\rm{rot} \mathbf{A} = b{\rm e}_3$. We choose
	$$\mathbf{A}=(-bx_2,0)\,.$$
	We  recall that
	\[\sigma_1=\begin{pmatrix}
		0&1\\
		1&0
	\end{pmatrix}\,,\quad \sigma_2=\begin{pmatrix}
		0&-i\\
		i&0
	\end{pmatrix}\,,\quad \sigma_3=\begin{pmatrix}
		1&0\\
		0&-1
	\end{pmatrix} \,.\]
	The study of the magnetic Dirac operator  on the half-plane with infinite mass boundary conditions was recently  carried forward in  \cite{barbaroux:hal-02889558}.
	Here we consider   general local conditions  at the edge  interpolating between zigzag and infinite mass. 
	More precisely: Let  $\gamma\in\R\cup\{+ \infty\}$, then we consider the
	self-adjoint realization $	\mathscr{D}_{\gamma}\equiv \mathscr{D}_{\gamma}(b)$ acting as \eqref{e1}
	on a subset of functions  $\psi=(\psi_1, \psi_2)\in L^2(\R^2_+, \C^2)$ satisfying, for all $x_1\in\R$, 
	\begin{equation}\label{bound.cond}\begin{cases}
			\psi_2(x_1,0)=\gamma \psi_1(x_1,0)&\text{ if }\gamma\in\R\,,
			\\
			\psi_1(x_1,0) = 0&\text{ if }\gamma = +\infty\,.
	\end{cases}\end{equation} 
	The two cases $\gamma\in\{0,+\infty\}$ are called zigzag, while  $\gamma=\pm1$ corresponds to infinite mass boundary conditions 
	(see Remark~\ref{rem-bc} bellow).
	\begin{remark}
		The domains of self-adjointness of the operators $\mathscr{D}_\gamma$  are already known: See
		\cite[Section 1C]{barbaroux:hal-02889558}  for the zigzag cases and \cite[Theorem 1.15]{barbaroux:hal-02889558} for  
		the infinite mass; the latter result can  be easily adapted for the non-zigzag cases.
		For the essential self-adjointness of $\mathscr{D}_\gamma$ on the class of Schwartz functions with infinite mass boundary conditions  
		see  \cite[Proposition 1.1]{Cornean_2023}, and for the  general case see  \cite{Paper2}.
	\end{remark}
	We denote by $X_1$ the operator of multiplication with $x_1$, and  by $J_1$ the current density operator, we have 
	$$J_1=-i[\mathscr{D}_\gamma, X_1]=-\sigma_1\,.$$
	Recall that the Landau Hamiltonian $\mathscr{D}_{bulk}(b)$ acts 
	on the whole plane as in \eqref{e1}  and its spectrum consists of the Landau levels given by 
	$\{\pm \sqrt{2n|b|},\,n\in \N_0\}
	$. 
 
			We say that $F:\R\to\R$ is equal to $a\in \R$ near $x_0\in\R$ if there exists an open interval $I$ around $x_0$ where $F(x)=a, x\in I,$ holds. 
   
	Our main result is the following.
	
	\begin{theorem}\label{prop.bec}
		Let $b>0$ and let $\chi=\mathds{1}_{(0,1)}$ be the indicator  function of the interval $(0,1)$. Let  $\gamma\in\R\cup\{+\infty\}$. 
		 Let $F\in C_0^2(\R)$ be such  that it equals $1$ near $n\geq 1$ Landau levels, and $0$ near the others.
		Then,  the operator $\chi(X_1)J_1 F'(\mathscr{D}_{\gamma})$ is trace class and the edge Hall conductance is given by
		\begin{equation}\label{hc1}
			\begin{split} 2\pi\mathrm{Tr}\big (\chi(X_1)J_1 F'(\mathscr{D}_{\gamma})\big )&=\begin{cases}
					n-1&\text{ if }\gamma = 0 \text{ and } F \text{ equals }1 \text{ near } 0\,, 
					\\
					n+1&\text{ if }\gamma = +\infty \text{ and } F \text{ equals }1 \text{ near } 0\,,
					\\
					n& \text{otherwise}\, .
				\end{cases} 
			\end{split}
		\end{equation}
	\end{theorem}
	\noindent
	{\bf Comments:}
	\begin{enumerate}
		\item 	 The fact that the  left hand side of this formula can be interpreted as an edge conductance is explained for instance in  \cite{EG02}. 
		\item Let us make the connection with the bulk-edge correspondence. Let $F$ and $n$ be as in Theorem  \ref{prop.bec} and define the orthogonal projection $P_n=F(\mathscr{D}_{bulk}(b))$. In our case, $P_n$ contains exactly $n$ bulk Landau levels and one can show  that its Chern number equals $n$, which encapsulates the non-trivial topology of the bulk projection (see \emph{e.g.} \cite{Cornean_2023}). 
			On the other hand,  since
			$F'$ equals $0$ near the Landau levels, $F'(\mathscr{D}_{\gamma})$ only selects edge-states. 
		
  \item If $b<0$, the Chern number of $P_n$ becomes $-n$, and the right hand side of formula \eqref{hc1} must be multiplied by $-1$.  
		\item 	Our third alternative in \eqref{hc1} indicates that the bulk-edge correspondence should hold for all non-zigzag conditions, a result which in our case is confirmed by brute force, i.e. by direct computation and comparison with the bulk Chern number.  The general proof of this fact, under more general conditions than purely constant magnetic field, will be considered in \cite{Paper2}. At least for $\gamma=1$, this has been shown to be the case \cite{Cornean_2023}; for Schr\"odinger-like operators see 
		\cite{cornean2022general}.
		\item If we work with zigzag boundary conditions, and if the zero-energy bulk Landau level belongs to the projection $P_n$, then one of the first two alternatives in \eqref{hc1} occurs. Thus the bulk-edge correspondence does not hold in this case.  Such an anomaly has been previously observed in other continuous  models such as shallow-water waves \cite{graf2021topology, PhysRevResearch.2.013147}, and for "regularized" non-magnetic Dirac-like operators \cite{PhysRevResearch.2.013147}. The latter  are in fact second order differential operators, with boundary conditions that are incompatible with first-order self-adjoint differential operators.
	\end{enumerate}
	
	\begin{remark}[On the boundary condition]\label{rem-bc}
		General local boundary conditions for Dirac operators are usually written as  
		$(-i\sigma_3(\sigma\cdot\mathbf{n})\cos\eta+\sigma_3\sin\eta)\psi=\psi,$ on $\partial \R^2_+\,,$
		where $\eta\in\left[-\frac{\pi}{2},\frac{3\pi}{2}\right)$
		(see \emph{e.g.} \cite{Berry1987} and \cite{sa2017,barbaroux:hal-02889558}).
		In the present situation, we have $\mathbf{n}=-e_2$ and thus
		\[(\sigma_1\cos\eta+\sigma_3\sin\eta)\psi=\psi\,,\quad \text{ on }\partial \R^2_+\,.\]
		We obtain \eqref{bound.cond} by setting 
		$\gamma=\frac{\cos\eta}{1+\sin\eta}=\tan\left(\frac\pi4-\frac\eta2\right)$
		with the convention $\gamma=+\infty$ in the case $\eta=-\pi/2$
		\emph{i.e.} $\psi_1 =0$ on $\partial \R^2_+$.
	\end{remark}

	\subsection{The energy dispersion curves}\label{edc}
	Our proof of Theorem~\ref{prop.bec} requires certain knowledge on the 
		energy dispersion curves and their corresponding eigenfunctions. In what  
		follows   we present   a description of these  curves  for different values of the boundary parameter $\gamma$. The main technical issue here is the lack of semi-boundedness  of $\mathscr{D}_\gamma$. This  can however be treated using appropriate variational methods \cite{GS99,DES00,schimmer2020friedrichs}; we follow the approach proposed in \cite{barbaroux:hal-02889558}. Before presenting our main results in this context we  establish the basic framework.
	\subsubsection{Setting}
	In view of the translation invariance in the $x_1$-direction, we may use the partial Fourier transform to represent  $\mathscr{D}_\gamma(b)$ as a family of fiber operators  $\mathscr{D}_{\gamma,\xi}(b)\equiv \mathscr{D}_{\gamma,\xi}$, with $\xi\in \R$.  Indeed,  we have (see \emph{e.g.} \cite{barbaroux:hal-02889558}) 
	\[
	\mathscr{D}_{\gamma} = \int_\R^\oplus\mathscr{D}_{\gamma,\xi} \,{\rm d}\xi\,,
	\]
	where the $1d$ magnetic Dirac operator $\mathscr{D}_{\gamma,\xi}$ acts as
	\[\begin{pmatrix}
		0&d_{\xi}(b)\\
		d^\ad_{\xi}(b)&0
	\end{pmatrix}\equiv \begin{pmatrix}
		0&d_{\xi}\\
		d^\ad_{\xi}&0
	\end{pmatrix}\,,\] 
	with $d_{\xi}=\xi-\partial_2+bx_2$ and  $d^\ad_{\xi}=\xi+\partial_2+bx_2$. Its domain is given for $\gamma\notin \{0,+\infty\}$  as 
	\begin{align}\label{eq.bck}
		\mathrm{Dom}(\mathscr{D}_{\gamma,\xi})=\{\psi\in H^1(\R_+,\C^2) : x_2\psi\in L^2(\mathbb{R}_+)\,\,\, \mbox{and}\,\,\, \psi_2(0)=\gamma\psi_1(0)\}\,.
	\end{align}
	As for the zigzag cases, by denoting  
	$$
	B^1(\R_+)=\{\psi\in H^1(\R_+,\C^2) : x_2\psi\in L^2(\mathbb{R}_+)\}
	$$ we have
	\begin{align*}
		&\mathrm{Dom}(\mathscr{D}_{0,\xi})
		=\{u\in L^2({\R_+}) : d^\ad_\xi u\in L^2(\R_+) \}\times H_0^1(\R_+)\cap B^1(\R_+)\,,\\
		&\mathrm{Dom}(\mathscr{D}_{\infty,\xi})
		= H_0^1(\R_+)\cap B^1(\R_+)\times
		\{u\in L^2({\R_+}) : d_\xi u\in L^2(\R_+) \}\,.
	\end{align*}

	We have that (\emph{cf.} \cite{barbaroux:hal-02889558}), for  any $\gamma\in\R\cup\{+\infty\}$   and  $\xi\in\R$, the operator  $\mathscr{D}_{\gamma,\xi}$ is  self-adjoint and has compact resolvent 
	(for $\gamma \not=\{0,+\infty\}$ it follows directly from the compact embedding of $H^1$ in $L^2$).
	We write the spectrum of $\mathscr{D}_{\gamma,\xi}$ as the set 
	$\{-\vartheta_j^-(\gamma,\xi) : j\in \N\}\cup \{\vartheta_j^+(\gamma,\xi) : j\in \N\}$ such that
	\begin{align}\label{eq.not}
		\ldots \le-\vartheta_2^-(\gamma,\xi)\le -\vartheta_1^-(\gamma,\xi)<0\leq \vartheta_1^+(\gamma,\xi)\le \vartheta_2^+(\gamma,\xi)\le 
		\ldots
	\end{align}
	For a given  boundary condition $\gamma\in \R\cup \{+\infty\}$ the map $\R\ni \xi\mapsto \vartheta_n^\pm(\gamma,\xi)$ defines the energy dispersion relation.
	
	The  following two propositions are shown in Section~\ref{sec.prel}.
	\begin{proposition}\label{prop.elementary}
		Let  $\gamma\in\R\cup\{+\infty\}$ , $\xi\in\R$, $b\not=0$.
		We have
		\begin{enumerate}[\rm (i)]
			\item For all $n\geq 1$, the eigenvalues $\vartheta^\pm_n(\gamma,\xi)$ are simple.
			\item
			For all $n\geq 1$, $\xi\mapsto \vartheta_n^\pm(\gamma,\xi)$ and  $\gamma \mapsto \vartheta_n^\pm(\gamma,\xi)$ are real-analytic.	
			\item 
			$0$ belongs to the spectrum of $\mathscr{D}_{\gamma,\xi}$ iff $\gamma= 0$, in case $b>0$, or $\gamma= +\infty$, in case $b<0$ .
		\end{enumerate}
	\end{proposition}
	In order to study the zigzag case  we introduce further 
	the $1d$ fibers of a magnetic Dirichlet Pauli operator $\mathscr{H}_{\xi}(b)$ for $b>0$. It acts as $  -\partial_2^2 + (\xi + bx_2)^2 + b\,,$
	with
	\[\mathrm{Dom}(\mathscr{H}_{\xi}(b))=\{\psi\in H^2(\mathbb{R}_+;\mathbb{C}) : x_2^2\psi\in L^2(\mathbb{R}_+)\,, \psi(0)=0\}\,.\]	
	It is well-known \cite{edge-classic} that $\mathscr{H}_{\xi}(b)$ is self-adjoint  with compact resolvent and that its spectrum consists on simple eigenvalues $(\nu_{n}^{\mathrm{Dir}}(b,\xi))_{n\in \N}$ with $$
	2b<\nu_{1}^{\mathrm{Dir}}(b,\xi)<\nu_{2}^{\mathrm{Dir}}(b,\xi)<\ldots\,.
	$$
	
	The following statements are 
	well-known (see \cite{S95} and \cite{barbaroux2021semiclassical,barbaroux:hal-02889558}).
	We specialize in the case $b>0$ since otherwise one can use the charge conjugation symmetry described in Remark~\ref{rem.symm} below. 
	\begin{proposition}\label{prop.zz}
		Consider the case  $\gamma\in\{0,+\infty\}$. Then,  for all $\xi\in\R$ and $b\not=0$, the spectrum of $\mathscr{D}_{\gamma,\xi}$ is symmetric with respect to $0$. Moreover, for $b>0$, we have
		\[\begin{split}
			&\vartheta_n^+(+\infty,\xi)=\sqrt{\nu_{n}^{\mathrm{Dir}}(b,\xi)-2b}\quad(n\geq 1)\,,
		\end{split}\]
		\[\begin{split}
			&\vartheta_1^{+}(0,\xi)=0,\,\,\,\mbox{and}\,\,\,\,
			\vartheta_n^+(0,\xi)=\sqrt{\nu_{n-1}^{\mathrm{Dir}}(b,\xi)}\quad(n\geq 2)\,.
		\end{split}\]
	\end{proposition}
	\begin{figure}[ht!]
		\centering
		\begin{subfigure}[t]{0.47\textwidth}
			\centering
			\includegraphics[width=\textwidth]{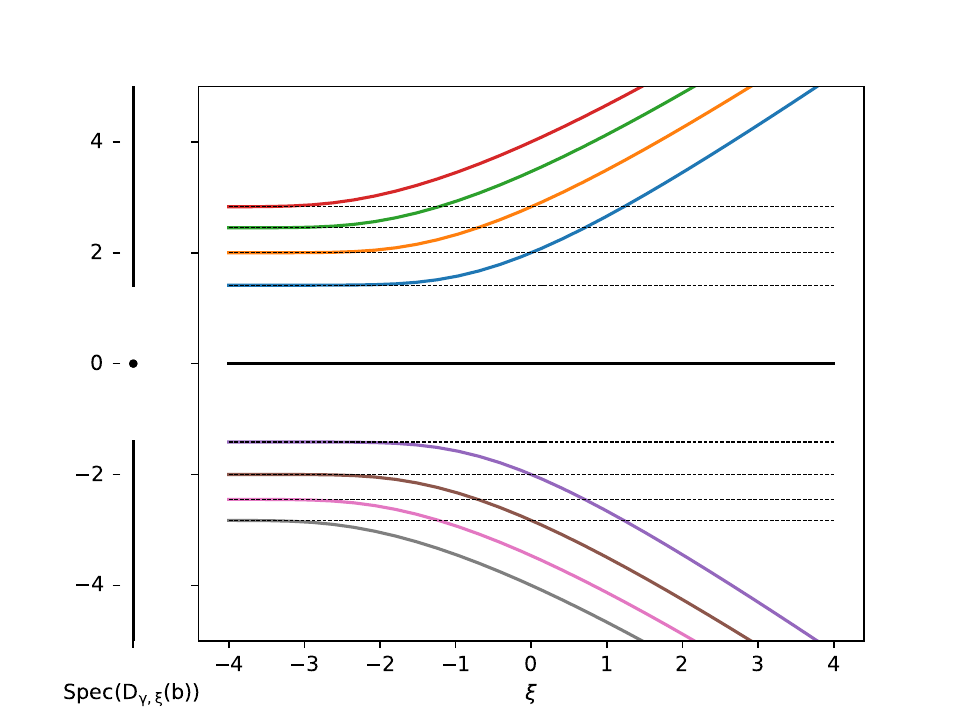}
			\caption{Case  $\gamma = 0$: Here $0$ is an eigenvalue of $\mathscr{D}_\gamma$ }
			\label{fig0}
		\end{subfigure}
		\hfill
		\begin{subfigure}[t]{0.47\textwidth}
			\centering
			\includegraphics[width=\textwidth]{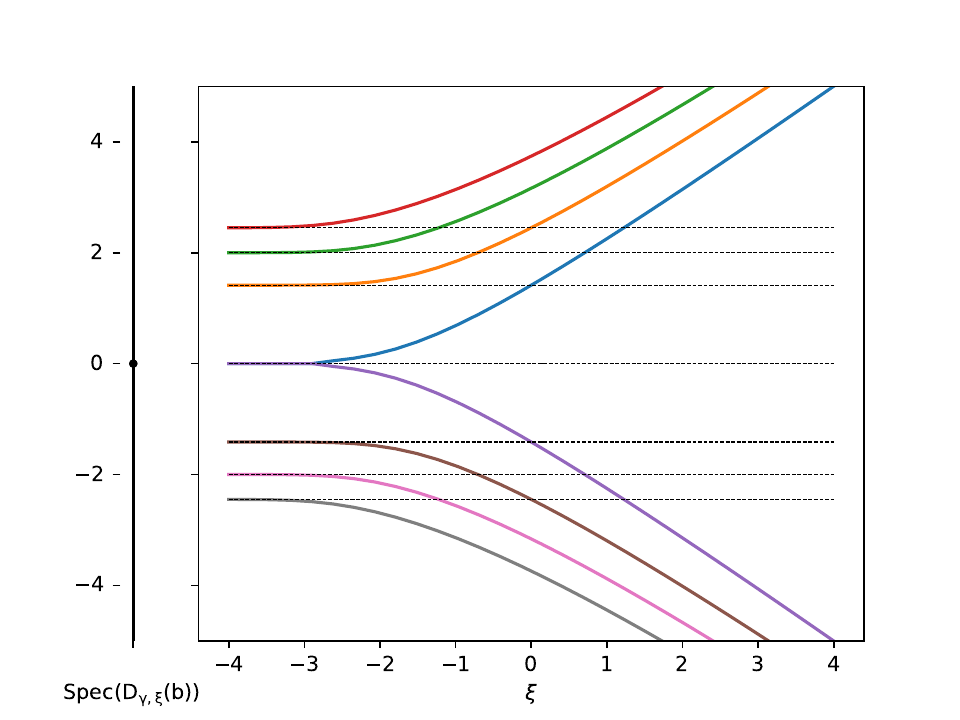}
			\caption{Case $\gamma = \infty$: Here $0$ is not an eigenvalue. Dashed lines correspond to the Landau levels of the bulk operator.}
			\label{figm1}
		\end{subfigure}
		\caption{Dispersion curves for the zigzag boundary conditions}
	\end{figure}
	\begin{remark}\label{rmk.limits}
		Recall \cite{edge-classic}  that,  for $n\geq 1$, the function $\nu_{n}^{\mathrm{Dir}}(b,\cdot)$ is increasing  and 
		\begin{equation*}
			\lim_{\xi\to-\infty}\nu_{n}^{\mathrm{Dir}}(b,\xi)=2nb\,,\quad \lim_{\xi\to+\infty}\nu_{n}^{\mathrm{Dir}}(b,\xi)=+\infty\,.
		\end{equation*}
	\end{remark} 
	Figures \ref{fig0} and \ref{figm1} give the dispersion curves of the zigzag Dirac operators and on the left of each figure, their spectrum. All illustrations presented in this article are obtained thanks to standard finite difference schemes, inverse power and Newton-like methods.
	
	\subsubsection{Main results for the energy dispersion curves}
	In view of the symmetry it is enough to consider the cases of positive magnetic field and non-negative boundary parameter {\emph i.e.} the case  $(b,\gamma)\in(0,+\infty)\times[0,+\infty]$ (see Remark~\ref{rem.symm}). 

	The following theorem gives a description of the dispersion curves when $\gamma\in(0,+\infty)$ and generalizes  the result obtained in \cite{barbaroux:hal-02889558} for $\gamma=1$.
	\begin{theorem}\label{thm.main}
		Let $\gamma\in(0,+\infty)$, $\xi\in\R$, $b>0$.
		The spectrum of $\mathscr{D}_{\gamma,\xi}(b)$ can be described as follows. Let $n\geq 1$.
		\begin{enumerate}[\rm (i)]
			\item The function $\vartheta^+_n(\gamma,\cdot)$ is increasing and
			\[\lim_{\xi\to-\infty}\vartheta^+_n(\gamma,\xi)=\sqrt{2(n-1)b}\,,\quad \lim_{\xi\to+\infty}\vartheta^+_n(\gamma,\xi)=+\infty\,.\]
			\item	The function $\vartheta^-_n(\gamma,\cdot)$ has a unique critical point, which is a non-degenerate minimum, and
			\begin{equation}\label{eq.limitheta-}	 
				\lim_{\xi\to-\infty}\vartheta^-_n(\gamma,\xi)=\sqrt{2nb}\,,\quad \lim_{\xi\to+\infty}\vartheta^-_n(\gamma,\xi)=+\infty\,.
			\end{equation}
		\end{enumerate}
	\end{theorem}
	\begin{figure}[ht!]
		\centering
		\includegraphics[width=0.5\textwidth]{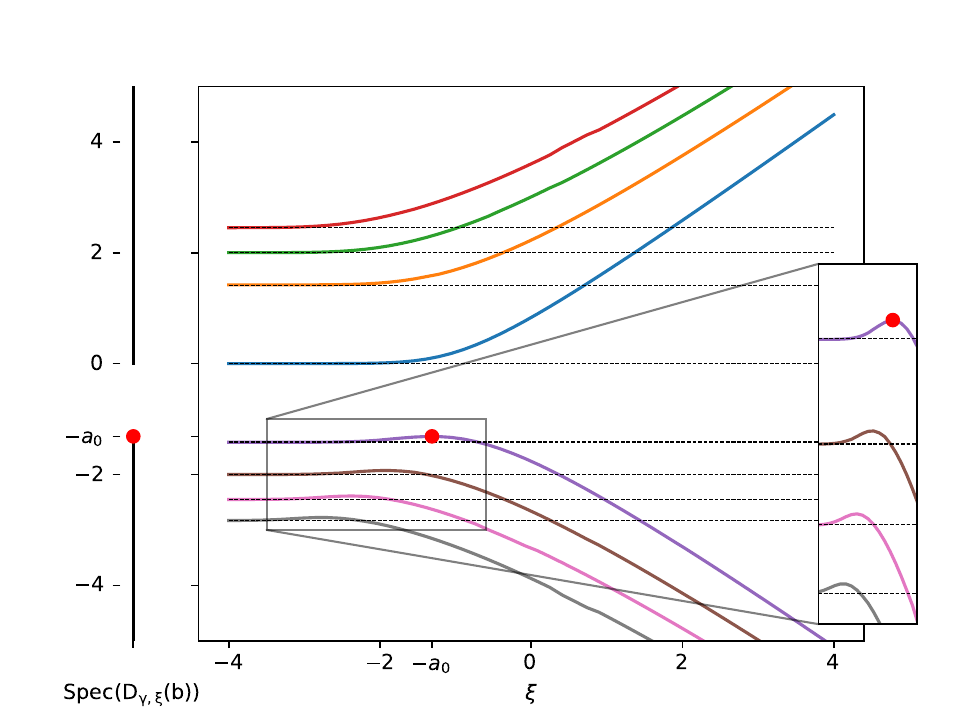}
		\caption{The dispersion curves of $\mathscr{D}_{\gamma,\xi}$ for $\gamma = b= 1$}
		\label{fig1}
	\end{figure}
	
	Figure \ref{fig1} gives the dispersion curves of the infinite mass Dirac magnetic operator with a special focus on the global maxima of the negative dispersion curves. Here 	$a_0\in(0,\sqrt{2})$, was introduced in \cite{barbaroux:hal-02889558}, it is the minimum of $\vartheta^-_1(1,\cdot)$, \emph{i.e.}, it is the size of the spectral gap of the Dirac operator with infinite mass boundary condition. The spectral gap of $\mathscr{D}_\gamma$ as a function of $\gamma$ can be read off from Figure~\ref{fig-gap}.

\begin{figure}[H]
		\centering
		\includegraphics[width=0.5\textwidth]{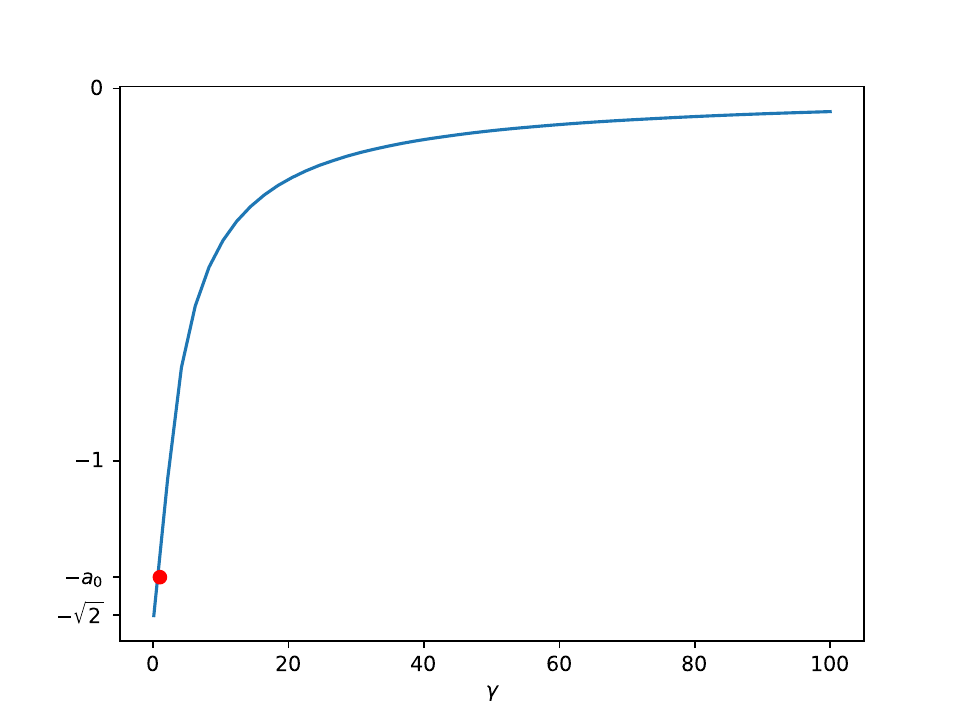}
		\caption{The value of the maximal negative energy of the full operator as a function of $\gamma$ for $b=1$. }
		\label{fig-gap}
	\end{figure}

\begin{figure}[ht!]
		\centering
		\includegraphics[width=0.5\textwidth]{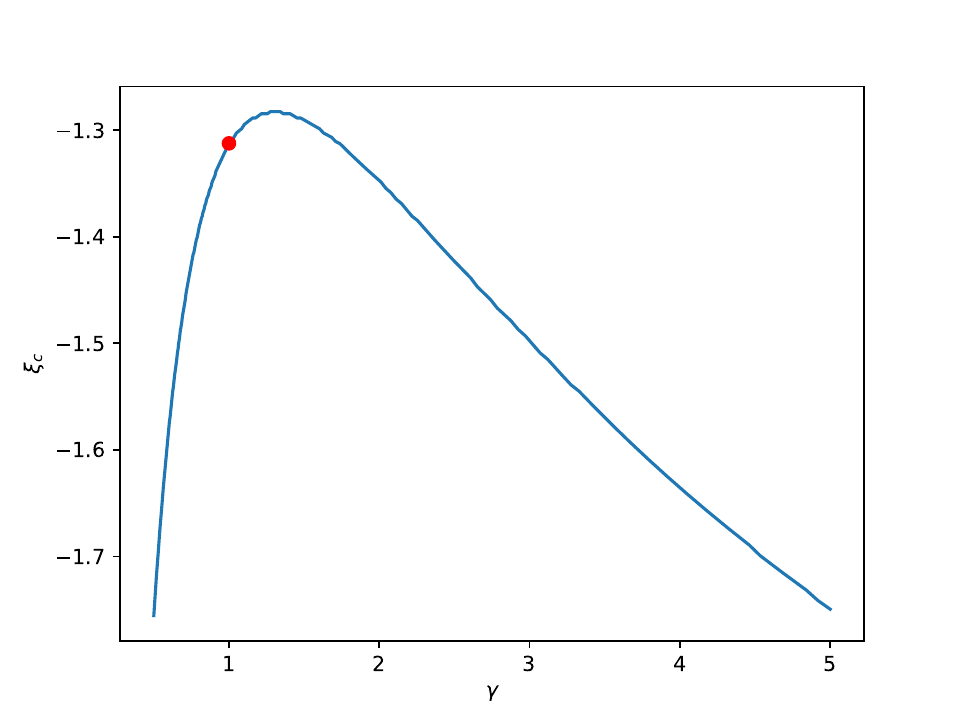}
		\caption{The location of the critical point of the first negative dispersion curve for $b= 1$. The red bullet refers to $\gamma=1$. }
		\label{fig-xi}
	\end{figure}

	Our next result describes   the dispersion curves as functions of  $\gamma $ and their zigzag limits.
	\begin{theorem}\label{prop.limitgamma}
		Let $b>0$. The families of functions $[0,+\infty)\ni\gamma\mapsto \vartheta_n^+(\gamma,\cdot)$ and $[0,+\infty)\ni\gamma\mapsto -\vartheta_n^-(\gamma,\cdot)$ are increasing with $\gamma$. Moreover, for all $n\geq 1$ and all $\xi\in\mathbb{R}$, we have
		\[\lim_{\gamma\to 0}\vartheta^\pm_n(\gamma,\xi)=\vartheta_n^\pm(0,\xi)\,,\quad \lim_{\gamma\to +\infty}\vartheta^\pm_n(\gamma,\xi)=\vartheta_n^\pm(+\infty,\xi)\,. \]	 
	\end{theorem}
	\begin{figure}
		\centering
		\includegraphics[width=1.\textwidth]{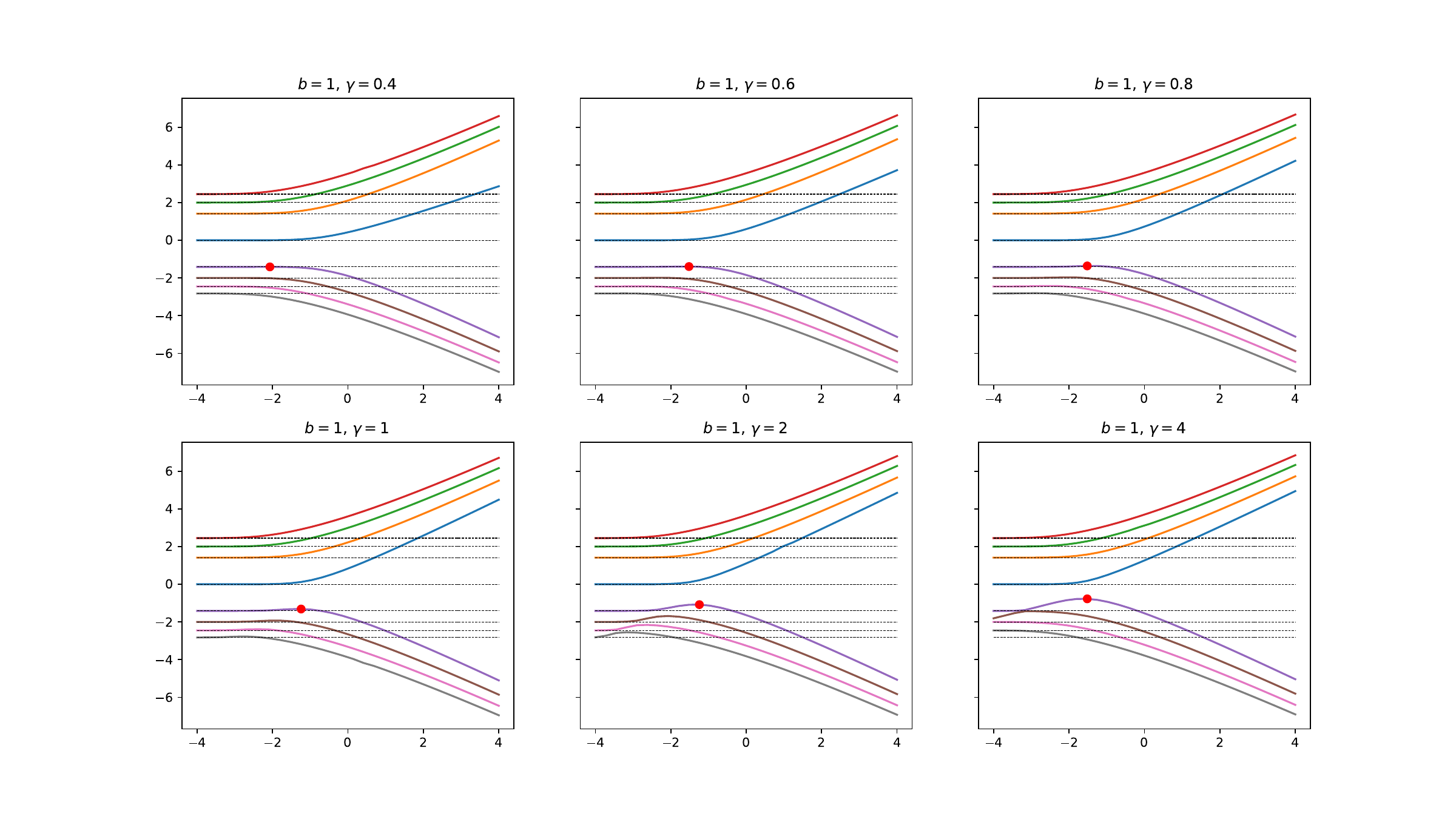}
		\caption{The dispersion curves for $b=1$ and various $\gamma$}
		\label{fig2}
	\end{figure}
	In Figure \ref{fig2} we present various pictures of the dispersion curves with varying $\gamma\in(0,+\infty)$. Moreover, 
	Figure \ref{fig3} illustrates the action of the symmetries, described in the following.
	\begin{remark}[Symmetries]\label{rem.symm}
		In view of the underlying symmetries it is enough to study the dispersion curves when $(b,\gamma)\in(0,+\infty)\times[0,+\infty]$. Indeed, 
		in order to also consider $\gamma<0$, we notice that
		\[\sigma_3\mathscr{D}_\gamma(b)=-\mathscr{D}_{-\gamma}(b)\sigma_3\,.\]
		Moreover, for $b<0$ we used  the charge conjugation $C\psi=\sigma_1\overline{\psi}$ which turns the boundary conditions into $\psi_2=\gamma^{-1}\psi_1$ and hence
		\[C\mathscr{D}_{\gamma}(b)=-\mathscr{D}_{\gamma^{-1}
		}{(-b)} \,C\,.\]
		For the fiber operators this leads to:
		\begin{equation}\label{eq.symfiber}
			C\mathscr{D}_{\gamma,\xi}(b)=-\mathscr{D}_{\gamma^{-1},-\xi}(-b)\,C\,,\quad \sigma_3\, \mathscr{D}_{\gamma,\xi}(b)=-\mathscr{D}_{-\gamma,\xi}(b)\,\sigma_3\, .
		\end{equation}
		In particular, we obtain the dispersion curves when $b<0$ from the curves when $b>0$ by changing $\gamma$ into $\gamma^{-1}$ and $\xi$ into $-\xi$.
	\end{remark}
	\begin{figure}[ht!]
		\centering
		\includegraphics[width=0.65\textwidth]{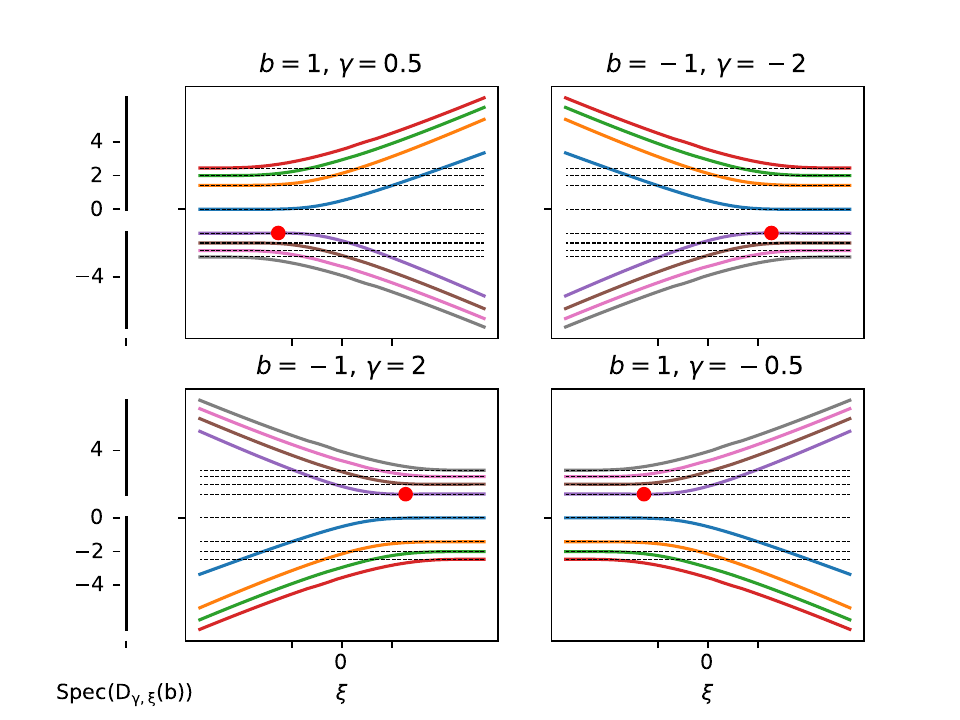}
		\caption{Action of the symmetries on the dispersion curves.}
		\label{fig3}
	\end{figure}
	
	\newpage
	{\it Organization of this article.} 
	In Section~\ref{sec.bulk} we prove  Theorem~\ref{prop.bec}. Sections \ref{sec.elem}, \ref{sec.main-dispersion}, and \ref{sec.proofs} are devoted to the description of the energy dispersion relations $\vartheta_n^\pm(\gamma,\xi)$.  In Section~\ref{sec.elem} we prove Propositions~\ref{prop.elementary} and \ref{prop.zz}.
	In addition, we state   in Theorem~\ref{prop.nupmlambda2}
	a fixed-point characterization of  $\vartheta_n^\pm(\gamma,\xi)$ in terms of a family   $\big(\nu_{n}^\pm(\alpha,\xi)\big)_{\alpha>0}$ of  the eigenvalues of   certain magnetic Schr\"odinger-like operators with Robin boundary conditions.
	We give a proof of this characterization  in  Section~\ref{sec.proofs}.
	In  Section~\ref{sec.main-dispersion} we investigate the fundamental
	mapping properties of  $\nu_{n}^\pm(\alpha,\xi)$ for $\alpha>0$ and $\xi\in \R$.
	Finally, in Section~\ref{sec.proofs}, we apply Theorem~\ref{prop.nupmlambda2}, together with the analysis of Section~\ref{sec.main-dispersion}, to prove Theorems~\ref{thm.main} and \ref{prop.limitgamma}.

	\section{Edge conductance  formula}\label{sec.bulk}
	In this section we prove Theorem~\ref{prop.bec}. 
	In doing so we use various  results on the energy dispersion curves which are stated in Section~\ref{edc} and proved in the next sections. 
	
	We let, for all $j\geq 1$, 
	\[\lambda_j(\xi)=\vartheta^+_j(\gamma,\xi)\,,\]
	and, for all $j\leq -1$,
	\[\lambda_j(\xi)=-\vartheta^-_j(\gamma,\xi)\,.\]
	 Let   $(\Psi_{\xi,j})_{\xi\in \R} $ be an analytic family of  normalized eigenfunctions of  $\mathscr{D}_{\gamma,\xi} $  associated with  
		$\lambda_j(\xi)$ (see Proposition~\ref{prop.elementary}).

		In view of the  asymptotic  behavior of $\lambda_j(\xi)$ as $ \xi\to\pm\infty$ -- stated in Theorem~\ref{thm.main} (for the non-zigzag case) and  Proposition~\ref{prop.zz}  and Remark~\ref{rmk.limits} (for the zigzag case) -- the proof of Theorem~\ref{prop.bec} reduces to showing  the following result.
	\begin{proposition}\label{prop-be}
		We let $\chi=\mathds{1}_{(0,1)}$ and consider $\gamma\in\R\cup\{+\infty\}$. Let us consider a function  $f\in C^1_0(\R)$  being zero near the Landau levels. Then, the operator $\chi(X_1)\sigma_1 f(\mathscr{D}_{\gamma})$ is trace class and there exists a finite $J\subset\mathbb{Z}\setminus\{0\}$ such that 
		\[2\pi\mathrm{Tr}\big(\chi(X_1)\sigma_1f(\mathscr{D}_{\gamma})\big)=\sum_{j\in J}\int_{\mathbb{R}}f(\lambda_j(\xi))\lambda'_j(\xi)\dd\xi\,.\]
		In particular, if $n\in\mathbb{N}$ and   $F\in C_0^2(\R)$  equals $1$ near $n$ Landau levels and $0$ near the others, we have
		\begin{equation}\label{hc2}
			-2i\pi\mathrm{Tr}\left(\chi(X_1)[\mathscr{D}_{\gamma},X_1]F'(\mathscr{D}_{\gamma})\right)=\sum_{j\in J}\Big (F\big (\lambda_j(-\infty)\big )-F\big (\lambda_j(+\infty)\big )\Big )\,.
		\end{equation}
	\end{proposition}
		Let us first state two useful elementary results.
		The first one is a  direct  consequence of Proposition \ref{prop.elementary} and Theorem \ref{thm.main}.
	\begin{lemma}\label{lem.truncLambda}
		Let us consider a  function $f\in C_0(\R)$  being zero near the Landau levels $\{\pm\sqrt{2n}\,,n\in \N\}$. Then, there exists a finite $J\subset\mathbb{Z}\setminus\{0\}$ such that for all $k\notin J$, $f\circ\lambda_k=0$ and for all $j\in J$, the functions $\xi\mapsto f(\lambda_j(\xi))$ have compact supports.
	\end{lemma}
	
	The following result might be elementary. We give its proof for the reader's convenience.
	\begin{lemma}\label{lemma.t-trace}
		Consider an operator $T$ on $L^2(\R^2_+)$  given as a Bochner integral (on a finite interval) of a continuous family of rank one operators 
		\begin{equation}\label{eq.T}
			T=\int_a^b  |\psi(\xi)\rangle \langle \phi(\xi)| d\xi\,.
		\end{equation}
		Assume further that the trace norm  $\Vert\psi(\xi)\Vert \,\Vert \phi(\xi)\Vert $ of the above integrand is uniformly bounded on $[a,b]$. 
		Then, $T$ is trace class and 
		\begin{align}\label{thet}
			{\rm Tr}(T)=\int_a^b {\rm Tr} \,  |\psi(\xi)\rangle \langle \phi(\xi)|	
			\, d\xi\,
			=\int_a^b \langle \phi(\xi),\psi(\xi)\rangle \, d\xi\,.
		\end{align}
	\end{lemma}
		\begin{proof}
			The integral in \eqref{eq.T} can be seen as a limit of a Riemann sum $T_n$,  which a priori only converges  in the operator norm topology. The integrand is a rank-one trace class operator, with a trace norm which is uniformly  bounded in $\xi$ on $[a,b]$. Hence, the trace norm of the $T_n$'s  is uniformly bounded in $n$. By  Lemma~\ref{tr-cl} from the Appendix we see that $T$,  which a priori is only a compact operator, is actually trace class.
			
			Let $\{f_j\}_{j\geq 1}$ be any orthonormal system. Then for all $N\geq 1$ we have
			\begin{align}\label{limit}
				\sum_{j=1}^N \langle f_j,T\, f_j\rangle= \int_a^b \sum_{j=1}^N \langle f_j, \psi(\xi)\rangle \, \langle \phi(\xi), f_j\rangle \, d\xi.
			\end{align}
			Using Cauchy-Schwarz and Bessel inequalities we get for every $N$:
			$$\Big |\sum_{j=1}^N \langle f_j, \psi(\xi)\rangle \, \langle \phi(\xi), f_j\rangle \Big |\leq \sqrt{\sum_{j=1}^N |\langle f_j,\psi(\xi)\rangle|^2}\, \sqrt{\sum_{j=1}^N |\langle f_j,\phi(\xi)\rangle|^2}\leq \Vert\psi(\xi)\Vert\, \Vert\phi(\xi)\Vert.$$
			By Lebesgue's dominated convergence theorem we
			can take $N\to\infty$ in \eqref{limit} to get \eqref{thet}. 
		\end{proof}
		\begin{proof}[Proof of Proposition~\ref{prop-be}]
			Let us first investigate the integral kernel of $f(\mathscr{D}_{\gamma})$ for some $f\in C_0^1(\R)$. 
			We can write
			\[\Big (f(\mathscr{D}_{\gamma})\psi\Big ) (x)=\int_{\R^2_{+}}K_f(x,x')\psi(x')\dd x'\,,\]
			with
			\[K_f(x,x')=\frac{1}{2\pi}\int_{\R}\dd\xi e^{i(x_1-x'_1)\xi}k_f(\xi, x_2,x'_2)\,,\]
			where
			\[k_f(\xi,x_2,x'_2)=\sum_{j\in J}f(\lambda_j(\xi))\big |\Psi_{\xi,j}(x_2) \big \rangle \, \big\langle \Psi_{\xi,j}(x'_2)\big |\,,\]
			and $J\subset\mathbb{Z}\setminus\{0\}$ is the finite set from Lemma~\ref{lem.truncLambda}.
			The technical issue here is that, even if we multiply by $\chi\equiv\chi(X_1)$ from the left  we can not directly apply Lemma~\ref{lemma.t-trace} since the 
			function $(x_1,x_2)\mapsto e^{ix_1\xi}\Psi_{\xi,j}(x_2)$ is not square integrable on $\R_+^2$.
			However, we observe  that 
			\begin{equation}\label{eq.comm}
				\sigma_1 \chi f(\mathscr{D}_{\gamma})=\sigma_1 \chi (1+iX_1) f(\mathscr{D}_{\gamma})(1+iX_1)^{-1}+\sigma_1 \chi (1+iX_1) 
				\big[     (1+iX_1)^{-1},
				f(\mathscr{D}_{\gamma}) 
				\big]\,.
			\end{equation}
			The first operator  above  can be written as 
			\begin{equation}\label{com1}
				\begin{split}
					\sigma_1 & \chi (1+iX_1) f(\mathscr{D}_{\gamma})(1+iX_1)^{-1}\\
					&=\frac{1}{2\pi}\sum_{j\in J}\int_{\R} f(\lambda_j(\xi))\, \big |(1+iX_1)\chi\, e^{iX_1\xi}\Psi_{\xi,j} \big \rangle\, \big \langle (1-iX_1)^{-1}\, e^{iX_1\xi}\Psi_{\xi,j}\big |\, \dd\xi\,.
				\end{split}
			\end{equation}
			The second operator has  a commutator term $[\cdot,\cdot]$  which can be explicitly computed using the following identity: We get, by doing  partial integration,  
			\begin{align*}
				2\pi \, K_f(x,x')\big (1+i\, x_1'\big )&=\int_{\R_\xi}\dd\xi \, \Big ((1+ix_1-\partial_\xi)\, e^{i(x_1-x'_1)\xi}\Big ) k_f(\xi, x_2,x'_2)\\
				&=2\pi \, \big (1+i\, x_1\big )\, K_f(x,x')+\int_{\R}\dd\xi \,  e^{i(x_1-x'_1)\xi}\, \partial_\xi k_f(\xi, x_2,x'_2)\,.
			\end{align*}
			Therefore, we get as operators  on $L^2(\R^2_+)$
			\begin{equation}\label{com3}
				\begin{split}
					\sigma_1 \chi &(1+iX_1) 
					\big[     (1+iX_1)^{-1},
					f(\mathscr{D}_{\gamma}) 
					\big]\\&=\frac{1}{2\pi}\sum_{j\in J}\int_{\R} f'(\lambda_j(\xi))\, \lambda_j'(\xi)\, \big | \sigma_1\chi\, e^{iX_1\xi}\Psi_{\xi,j} \big \rangle\, \big \langle (1-iX_1)^{-1}\, e^{iX_1\xi}\Psi_{\xi,j}\big |\, \dd\xi\\
					&+\frac{1}{2\pi}\sum_{j\in J}\int_{\R} f(\lambda_j(\xi))\, \big |\sigma_1\chi\, e^{iX_1\xi}\big (\partial_\xi \Psi_{\xi,j}\big ) \big \rangle\, \big \langle (1-iX_1)^{-1}\, e^{iX_1\xi}\Psi_{\xi,j}\big |\, \dd\xi\\
					&+\frac{1}{2\pi}\sum_{j\in J}\int_{\R} f(\lambda_j(\xi))\, \big |\sigma_1 \chi \, e^{iX_1\xi} \Psi_{\xi,j} \big \rangle\, \big \langle
					(1-iX_1)^{-1}\, e^{iX_1\xi}\, \big (\partial_\xi\Psi_{\xi,j}\big )\big |\, \dd\xi\,.
				\end{split}
			\end{equation}
			Notice that thanks to Lemma~\ref{lem.truncLambda} the integrals above take place on a finite interval.
			Therefore, each of the four terms appearing in \eqref{com1} and \eqref{com3}  can be seen as Bochner integrals involving rank one operators in $L^2(\R^2_+)$ whose trace is uniformly bounded on compact sets. Hence,  Lemma~\ref{lemma.t-trace}
			can be applied to each of the terms involved in \eqref{eq.comm}. In particular, as a finite sum of trace class operators, $\sigma_1 \chi f(\mathscr{D}_{\gamma}) $ is trace class. A quick computation using  Lemma~\ref{lemma.t-trace} for each term in \eqref{com3} gives
			\begin{align*}
				{\rm Tr}\Big( \sigma_1 \chi (1+iX_1) 
				& \big[     (1+iX_1)^{-1},
				f(\mathscr{D}_{\gamma}) \big]\Big)\\&= 
				\frac{1}{2\pi}\sum_{j\in J}\int_{\R} \partial_\xi \Big(
				f(\lambda_j(\xi)) 
				\big \langle \sigma_1\chi \Psi_{\xi,j}, 
				(1-iX_1)^{-1}\,\Psi_{\xi,j}
				\big \rangle
				\Big)
				\, \dd\xi\,=0\,,
			\end{align*}
			where in the last step we used that the term $\partial_\xi(\dots)$ has compact support in $\xi$. Thus, we get 
			\begin{align*}
				&	\mathrm{Tr}\big(\chi\sigma_1f(\mathscr{D}_{\gamma})\big)=
				\mathrm{Tr}\big(
				\sigma_1 \chi (1+iX_1) f(\mathscr{D}_{\gamma})(1+iX_1)^{-1}
				\big)\\
				&=\frac{1}{2\pi}\sum_{j\in J}\int_{\R}f(\lambda_j(\xi)) \left\langle\sigma_1\Psi_{\xi,j},  \Psi_{\xi,j}\right\rangle_{L^2(\R_+,\C^2)} \dd\xi	\,.
			\end{align*}
			Now  the conclusion follows since $\left\langle\sigma_1\Psi_{\xi,j},  \Psi_{\xi,j}\right\rangle_{L^2(\R_+,\C^2)}=\left\langle\partial_\xi(\mathscr{D}_{\gamma,\xi})\Psi_{\xi,j}, \Psi_{\xi,j}\right\rangle_{L^2(\R_+,\C^2)}$ which, by the Feynman-Hellmann  theorem, equals $\lambda_j'(\xi)$. In particular, we get \eqref{hc2} by  writing $F'=f$ and integrating in $\xi$.
		\end{proof}

		\section{Energy dispersion curves}\label{sec.elem}
		We  start this section  by showing Propositions~\ref{prop.elementary} and \ref{prop.zz}. They state the basic properties of the 
		solutions of the eigenvalue problem, for $\xi\in \R$ and  $(b,\gamma)\in (0,+\infty)\times [0,+\infty]$
		\begin{align}\label{eq.e-v}
			\mathscr{D}_{\gamma,\xi} u=\lambda u \,.
		\end{align}
		We show that these solutions are related to a Schr\"odinger-like  problem with Robin boundary conditions.
  For zigzag boundary conditions this property is already clear from Proposition~\ref{prop.zz}. For $\gamma\in (0,+\infty)$ we establish this relation in   Lemma~\ref{lem.bij} below. 
  
  Moreover, we present in Theorem~\ref{prop.nupmlambda2} a characterization
		of the eigenvalues 
		$\vartheta_n^\pm(\gamma,\xi)$ in terms of a fixed-point problem that runs along a family of eigenvalues  
		$\big(\nu^\pm(\alpha,\xi)\big)_{\alpha>0}$
		of certain Schr\"odinger-like operators.   
		\subsection{Preliminaries}\label{sec.prel}
		Let us  investigate  some preliminary facts. (Throughout this paragraph we assume $b>0$.) The eigenvalue equation
		\eqref{eq.e-v}  can be rewritten as
		\begin{equation*}
			d_\xi^\ad u_1=\lambda u_2\,,\quad d_\xi u_2=\lambda u_1\,.
		\end{equation*}
		Then, we have $d_\xi d_\xi^\ad u_1=\lambda^2 u_1$ and $d^\ad_\xi d_\xi u_2=\lambda^2 u_2$. Moreover, from the classical theory of ODEs, we see that $u_1$ and $u_2$ are smooth on $[0,+\infty)$. Since $u_2(0)=\gamma u_1(0)$ (or $u_1(0)=0$ when $\gamma=+\infty$) we obtain Robin-type boundary conditions for $u_1$ and $u_2$, separately.  Thus,   \eqref{eq.e-v} implies 
		\begin{align}\label{eq.eve2}
			&d_\xi d_\xi^\ad u_1=\lambda^2 u_1\,\quad \big(d^\ad_\xi u_1(0)=\gamma\lambda u_1(0)\big )\,,\\\label{eq.eve2.1}
			& d^\ad_\xi d_\xi u_2=\lambda^2 u_2\,\quad\big(
			d_\xi u_2(0)= \tfrac{\lambda}{\gamma} u_2(0)\big)\,.
		\end{align}
		\begin{proof}[Proof of Proposition \ref{prop.elementary}]
			Let us consider the eigenvalue equations  \eqref{eq.eve2} and \eqref{eq.eve2.1}.
			From the standard theory of  initial value problems, we see that  $u_j$ belongs to a space of dimension at most $1$. Therefore, $\dim{\ker}\big (\mathscr{D}_{\gamma,\xi}-\lambda\big )\leq1$. This proves the simplicity of the non-zero eigenvalues.
			
			Let us now discuss the existence of zero modes. For $\lambda=0$, we have $ d_\xi u_2=0$ so that $u_2$ is proportional to $e^{\frac{1}{2b}(\xi+bx_2)^2}$, which is not in $L^2(\R_+)$ implying that  $u_2=0$ holds. Moreover, we also check that $\gamma u_1(0)=0$. 
			Using   $d_\xi^\ad u_1=0$ we see that $u_1$ is proportional to $e^{-\frac{1}{2b}(\xi+bx_2)^2}$,  which belongs to  $L^2(\R_+)$ but
			it does not vanish at $x_2=0$. Therefore,  we find that $u_1=0$ unless  $\gamma= 0$.
			
			The family  $(\mathscr{D}_{\gamma,\xi})_{\xi\in\R}$ being analytic of type $(A)$ (in the Kato sense, see \cite{ReedSimon1978}), the simplicity of the eigenvalues implies their analyticity.	
		\end{proof}
		Next we discuss  the zigzag operators \emph{i.e.} the cases  $\gamma\in\{0,+\infty\}$. 
		\begin{proof}[Proof of Proposition~\ref{prop.zz}]
			The fact that the eigenvalues are symmetric with respect to zero follows from  \eqref{eq.symfiber}, hence, we may look at the non-negative ones only. 
			
			Let us consider the case $\gamma=0$ \emph{i.e.} $u_2(0)=0$. As we have just seen, we have a zero mode and, with our convention, we have $\vartheta^+_1(0,\xi)=0$. Let us describe the non-zero eigenvalues. Let $\lambda$ be a positive eigenvalue and $u$ a corresponding eigenfunction. In view of \eqref{eq.eve2}, we see that $u_2$ cannot be $0$ and it is an eigenfunction of $$d_\xi^\ad d_\xi=-\partial_2^2+(\xi+bx_2)^2+b\,,$$ with Dirichlet condition. In particular, $\lambda^2$ belongs to the spectrum of $d_\xi^\ad d_\xi$ with Dirichlet condition. Conversely, if $\mu>0$ is an eigenvalue of this operator, we write $d_\xi^\ad d_\xi v=\mu v$ with $v(0)=0$ and we let $u=\mu^{-\frac12}d_\xi v \,\,(\neq 0)$ and we have
			\[d_\xi^\ad  u=\sqrt{\mu} v\,,\quad d_\xi v=\sqrt{\mu} u\,,\quad v(0)=0\,,\]
			which means that $\sqrt{\mu}$ is an eigenvalue of $\mathscr{D}_{0,\xi}$.
			
			The case $\gamma=+\infty$ is quite similar although we have no zero modes. Now, $u_1$ is an eigenfunction (with eigenvalue $\lambda^2$) of $$d_\xi d_\xi^\ad=-\partial_2^2+(\xi+bx_2)^2-b\,,$$ with Dirichlet condition. Conversely, consider an eigenvalue $\mu>0$ of this operator. Proceeding as before,  we write $d_\xi d^\ad_\xi u=\mu u$ with $u(0)=0$  and let $v=\mu^{-\frac12}d^\ad _\xi u\,\,(\neq 0)$ to get that  $\sqrt{\mu}$ is an eigenvalue of $\mathscr{D}_{\infty,\xi}$.
		\end{proof}

		\subsection{A characterization of the eigenvalues for the non-zigzag case}\label{sec.cara}
		We consider $\gamma\in (0,+\infty)$. 
		Let $\lambda\neq 0$ be an eigenvalue of $\mathscr{D}_{\gamma,\xi}(b)$. 
		Multiplying \eqref{eq.eve2} by $u_1$ and integrating by parts yields
		\[\langle d_\xi d_\xi^\ad  u_1,u_1\rangle=\|d^\ad _\xi u_1\|^2+d^\ad _\xi u_1(0)u_1(0)=\|d^\ad _\xi u_1\|^2+\lambda u_2(0)u_1(0)=\|d^\ad _\xi u_1\|^2+\lambda \gamma u^2_1(0)\,.\]
		Moreover, proceeding analogously for  the second component in  \eqref{eq.eve2.1} we get 
		\[\langle d_\xi^\ad  d_\xi u_2,u_2\rangle=\|d_\xi u_2\|^2-d_\xi u_2(0)u_2(0)
		=\|d_\xi u_2\|^2-\lambda u_1(0)u_2(0)=\|d_\xi u_2\|^2-\tfrac{\lambda}{\gamma} u^2_1(0)\,.\]
		This suggests to introduce the following
		family of quadratic forms. 
		\begin{definition}
			Let $\alpha>0$.  We define the auxiliary quadratic forms,  for  $u\in B^1(\R_+)$,  as
			\begin{equation}\label{def.qfalpha}
				\begin{split}
					q^+_{b,\alpha,\xi}(u)=\|d_\xi^\ad  u\|^2+\alpha u^2(0)\,,\\
					q^-_{b,\alpha,\xi}(u)=\|d_\xi u\|^2+\alpha u^2(0)\,.
				\end{split}
			\end{equation}
			They are both non-negative and closed.
			We denote by $\mathfrak{h}^{\pm}_{\alpha,\xi}$ the corresponding 
			self-adjoint Schr\"odinger operators. 
		\end{definition}
		\begin{remark}\label{rem.fried}
			By Friedrichs' extension theorem we have  that 
			$\mathrm{Dom}(\mathfrak{h}^\pm_{\alpha,\xi})\subset B^1(\R_+)$ and  for $u^\pm\in {\rm Dom}(\mathfrak{h}^\pm_{\alpha ,\xi})$ 
			\begin{align}
				&\mathfrak{h}^+_{\alpha,\xi} u^+= d_\xi d_\xi^\ad u^+\,, \quad 
				d^\ad_\xi u^+(0)=\alpha u^+(0)\,,\label{eq.dd1}\\
				&\mathfrak{h}^-_{\alpha,\xi} u^-
				=d_\xi^\ad d_\xi u^-\,,\quad 
				d_\xi u^-(0)=-\alpha u^-(0)\label{eq.dd2}\,.
			\end{align}
		\end{remark}
		\begin{remark}
			Integration by parts yields
			\begin{equation}
				\begin{split}
					q^+_{b,\alpha,\xi}(u)=\|u'\|^2+\|(\xi+b x_2)u\|^2-b\|u\|^2+(\alpha-\xi)u^2(0)\,,\\
					q^-_{b,\alpha,\xi}(u)=\|u'\|^2+\|(\xi+b x_2)u\|^2+b\|u\|^2+(\alpha+\xi)u^2(0)\,.
				\end{split}
			\end{equation}
			We also observe that  
   \begin{equation}\label{jmb1}
   q^-_{b,\alpha,\xi}=q^+_{-b,\alpha,-\xi},
   \end{equation}
   which reflects the first relation in \eqref{eq.symfiber}. In what follows, we drop the reference to $b$ in the notation.
		\end{remark}
		In relation  to our problem we see that, for $u=(u_1,u_2)$ an eigenfunction of $\mathscr{D}_{\gamma,\xi}$, we have
		\begin{align*}
			&q^+_{\lambda \gamma,\xi}(u_1)=\langle u_1, d_\xi d_\xi^\ad u_1\rangle =\lambda^2\|u_1\|^2\,,\quad \mbox{for $\lambda>0\quad$and,}\\
			&q^-_{-\lambda\gamma^{-1},\xi}(u_2)=\langle u_2, d_\xi^\ad d_\xi u_2\rangle =\lambda^2\|u_2\|^2\,,\quad 
			\mbox{for $\lambda<0\,.$}\
		\end{align*}
		Next, we describe a bijection existing between  the kernels of $\mathscr{D}_{\gamma,\xi}-\lambda$ and  $\mathfrak{h}^+_{\gamma\lambda,\xi}-{\lambda^2}$ provided $\gamma\lambda>0$.	If $\gamma\lambda<0$ analogous statements can be obtained for 
		$\mathfrak{h}^-_{-\lambda/\gamma,\xi}-{\lambda^2}$.
		The following lemma is a straightforward adaptation of \cite[Proposition 2.9]{barbaroux:hal-02889558}. We recall its proof for the convenience of the reader and we emphasize that it does not require sign assumptions on $b$ and $\gamma$.
		\begin{lemma}\label{lem.bij}
			Let $(b,\gamma)\in\mathbb{R}^2$, $\xi\in\R$. Then, for any 
			$\lambda\in\mathbb{R}\setminus\{0\}$ with $\gamma\lambda>0$, the map \[\mathscr{J} : \ker(\mathfrak{h}^+_{\gamma\lambda,\xi}-\lambda^2)\ni u\mapsto (u,\lambda^{-1}d^\ad_{\xi} u)\in\ker(\mathscr{D}_{\gamma,\xi}-\lambda)\] 
			is well-defined and it is an isomorphism.	
		\end{lemma}
		
		\begin{proof}
			First, let $u\in \ker(\mathfrak{h}^+_{\gamma\lambda,\xi}-\lambda^2)$ and $v\in  \mathrm{Dom}(\mathscr{D}_{\gamma,\xi}-\lambda)$. We have
			\[\langle \mathscr{J}(u), (\mathscr{D}_{\gamma,\xi}-\lambda)v\rangle= \langle u, d_\xi v_2-\lambda v_1\rangle+\lambda^{-1}\langle d_\xi^\ad   u,d^\ad_\xi v_1-\lambda v_2\rangle\,,\]
			so that, by integrating by parts,
			\[\langle \mathscr{J}(u), (\mathscr{D}_{\gamma,\xi}-\lambda)v\rangle= \lambda^{-1}(q^+_{\gamma\lambda,\xi}(u,v_1)-\lambda^2\langle u,v_1\rangle)=0\,,\]
			where we used that $u\in  \ker(\mathfrak{h}^+_{\gamma\lambda,\xi}-\lambda^2)$. Thus, $\mathscr{J}$ is well defined. It is also injective from the very definition. For the surjectivity, we consider $(u_1,u_2)\in\ker(\mathscr{D}_{\gamma,\xi}-\lambda)$. We have
			\[d_\xi u_2=\lambda u_1\,,\quad d_\xi^\ad u_1=\lambda u_2\,.\]
			We only have to check that $u_1\in  \ker(\mathfrak{h}^+_{\gamma\lambda,\xi}-\lambda^2)$. Take $v\in B^1(\R_+)$ and notice that
			\[\begin{split}
				q^+_{\gamma\lambda,\xi}(u_1,v)-\lambda^2\langle u_1,v\rangle&=\langle d_\xi^\ad u_1,d_\xi^\ad v\rangle+\gamma\lambda u_1(0)v(0)-\lambda^2\langle u_1,v\rangle\\
				&=\lambda\langle  u_2,d_\xi^\ad v\rangle+\gamma\lambda u_1(0)v(0)-\lambda^2\langle u_1,v\rangle\\
				&=\lambda\langle d_\xi u_2, v\rangle-\lambda^2\langle u_1,v\rangle\\
				&=0\,.
			\end{split}\]
			This finishes the proof.
		\end{proof}
		Let us now turn to the characterization. 
		Since the family  $(q^\pm_{\alpha,\xi})_{(\alpha,\xi)\in\R_+\times\R}$ is analytic  on the  common domain $B^1(\R_+)$ the  eigenvalues of $\mathfrak{h}^\pm_{\alpha,\xi}$ (which are all simple) are also real analytic with respect to $\alpha$ and to $\xi$. We denote them by $\nu^\pm_n(\alpha,\xi)$ so that
		$$
		0\le \nu^\pm_1(\alpha,\xi)< \nu^\pm_2(\alpha,\xi)< \dots
		$$
		The following result completely characterizes positive and negative eigenvalues of $\mathscr{D}_{\gamma,\xi}$ in terms of $\alpha\mapsto\nu^\pm_n(\alpha,\xi)$. 
		Recall the notation in \eqref{eq.not}.
		\begin{theorem}\label{prop.nupmlambda2}
			Let $b\neq 0$ and $\gamma\in(0,+\infty)$. The equation
			\begin{equation}\label{eq.nu+=lambda2}
				\nu^+_n(\gamma\lambda,\xi)=\lambda^2
			\end{equation}
			has a unique positive solution $\lambda=\vartheta^+_n(\gamma,\xi)$.
			Moreover,  the equation
			\begin{equation}\label{eq.nu-=lambda2}	
				\nu^-_n(\gamma^{-1}\lambda,\xi)=\lambda^2
			\end{equation}
			has a unique positive solution $\lambda=\vartheta^-_n(\gamma,\xi)$.
		\end{theorem}
  The proof of this theorem uses Lemma~\ref{lem.bij} and  requires the analysis of the auxiliary quadratic forms performed in  Section~\ref{sec.main-dispersion}; we postpone it to Section~\ref{sec.proofs}. 
		\section{The auxiliary quadratic forms}\label{sec.main-dispersion}
		In this section we perform a detailed study of the auxiliary quadratic forms from Definition~\ref{def.qfalpha}. 
		We restrict the analysis  to the case in which $(b,\gamma)\in(0,+\infty)\times(0,+\infty)$.
		
		For $\alpha>0$  and $\xi\in \R$ consider the eigenvalue problems
		(recall  Remark~\ref{rem.fried})
		\begin{align}\label{eq.eve}
			\mathfrak{h}^\pm _{\alpha,\xi} u^\pm_{\alpha,\xi}=\nu^\pm(\alpha,\xi) u^\pm_{\alpha,\xi}\,.
		\end{align}
		Most of the  following  results can be traced back to \cite{barbaroux:hal-02889558} (notice, however, the different convention for $q^-_{\alpha,\xi}$). For the sake of completeness we present  a concise argument. 
		\subsection{Study of $\alpha\mapsto \nu_n(\alpha,\xi)$}
		\begin{lemma}\label{lem.qinalpha}
			Let $\nu^{\pm}(\alpha,\xi)$ be an eigenvalue as in \eqref{eq.eve}. Then, 	the function  $\alpha\mapsto\nu^{\pm}(\alpha,\xi)$ is increasing
			and
			\begin{equation}\label{eq.dalphanu}
				\partial_\alpha\nu^\pm(\alpha,\xi)=(u_{\alpha,\xi}^\pm(0))^2>0\,.
			\end{equation}
		\end{lemma}
		\begin{proof}
			We only argue for the $+$ case. To simplify notation, we denote the corresponding normalized solution of \eqref{eq.eve} as $u \,\,(\equiv u^{+}_{\alpha,\xi})$ and we drop the reference to $\xi$ and $\alpha$ when not relevant.
			
			Let us first observe that for any smooth function $g$ on $[0,\infty)$ we have, integrating by parts,
			\begin{equation}\label{eq.intbyparts}
				\begin{split}
					\langle u,( d d^\ad -\nu) g\rangle
					&=
					\langle d^\ad u,  d^\ad g\rangle+u(0) (d^\ad g)(0)-\nu \langle u,  g\rangle\\
					&=
					u(0) (d^\ad g)(0)
					-( d^\ad u)(0)  g(0)\,.
				\end{split}
			\end{equation}
			In view of the smoothness of $u$ with respect to $\alpha$ and $x$, we see that $d^\ad \partial_\alpha u=d^\ad\partial_\alpha u$. Hence, since $d^\ad u(0)=\alpha u(0)$, we get that  
			\begin{align}\label{eq.dalpha}
				(d^\ad\partial_\alpha u)(0)=u(0)+\alpha \partial_\alpha u(0)\,.
			\end{align}	
			Taking derivative with respect to $\alpha$ in the  eigenvalue equation we get
			\begin{align*}
				(d d^\ad -\nu)\partial_\alpha u=(\partial_\alpha \nu) u\,.
			\end{align*}
			After multiplying by $u$ and  integrating  we use \eqref{eq.intbyparts} to get
			\begin{align*}
				\partial_\alpha \nu&=\langle u, d d^\ad \partial_\alpha u\rangle
				-\nu \langle u,  \partial_\alpha u\rangle=
				u(0) (d^\ad\partial_\alpha u)(0)-( d^\ad u)(0)  (\partial_\alpha u)(0)\,.
			\end{align*}
			The conclusion follows  using \eqref{eq.dd1} and \eqref{eq.dalpha}.
		\end{proof}
		\begin{lemma}\label{lem.q-a0}
			For all $n\geq 1$, we have
			\begin{align}\label{lim.q-}
				&	\lim_{\alpha\to 0}\nu^-_n(\alpha,\xi)=\nu_n^{-}(0,\xi)=\nu^{\mathrm{Dir}}_n(-b,-\xi)\,\\
				& 
				\lim_{\alpha\to+\infty}\nu^-_n(\alpha,\xi)=\nu_n^{\mathrm{Dir}}(b,\xi)\,.
				\label{eq.111}	\end{align}
		\end{lemma}
		\begin{proof}
			The first equality follows by analyticity.	Then, we have $q^{-}_{0,\xi}(u)=\|d_\xi u\|^2$.  The corresponding operator $\mathfrak{h}^-_{0,\xi}$ has no zero mode. Now, if $\nu$ is a positive eigenvalue, we have
			\[d_\xi^\ad d_\xi u=\nu u\,,\quad d_\xi u(0)=0\,.\]
			Letting $v=d_\xi u$, we get $d_\xi d_\xi^\ad  v=\nu v$ with $v(0)=0$. This shows that $\nu$ belongs to the spectrum of the Dirichlet realization of $d_\xi d_\xi^\ad $. Conversely, if $\nu>0$ is an eigenvalue of $d_\xi d_\xi^\ad $ associated with the eigenfunction $v$, we have
			\[d_\xi^\ad  d_\xi u=\nu u\,,\quad u=d_\xi^\ad v\,,\quad d_\xi u(0)=0\,.\]
			We can check that $u\neq 0$ (unless $v=0$). Thus, $\nu$ belongs to the spectrum of the operator  $\mathfrak{h}^-_{0,\xi}$.
			
			When $\alpha\to+\infty$, we are in a singular regime.
			By using that $H^1_0(\R_+)\cap B^1(\R_+)\subset B^1(\R_+)$ and the min-max principle, we see that 	
			\[\nu_n^{-}(\alpha,\xi)\leq \nu^{\mathrm{Dir}}_n(b,\xi)\,.\]
			Conversely, let us consider 
			\[E_n(\alpha,\xi):=\underset{1\leq k\leq n}{\mathrm{span}}\, u^-_{\alpha,\xi,k}\,.\]
			We notice that, for all $u\in E_n(\alpha,\xi)$,
			\[(\alpha+\xi)u^2(0)\leq q^-_{\alpha,\xi}(u)\leq \nu_n^{-}(\alpha,\xi)\|u\|^2\leq  \nu^{\mathrm{Dir}}_n(b,\xi)\|u\|^2\,,\]
			so that, for $\alpha$ large enough,
			\[u^2(0)\leq\frac{ \nu^{\mathrm{Dir}}_n(b,\xi)}{\alpha+\xi}\|u\|^2=\mathscr{O}(\alpha^{-1})\|u\|^2\,.\]
			Then, we also notice that
			\[\|u'\|^2+\|(\xi+b x_2)u\|^2+b\|u\|^2\leq q^-_{\alpha,\xi}(u)\leq \nu^{\mathrm{Dir}}_n(b,\xi)\|u\|^2\,.\]
			Let us consider a smooth cutoff function $\chi$ with compact support equal to $1$ near $0$. The function
			\[\tilde u(x_2)=u(x_2)-u(0)\chi(x_2)\]
			satisfies the Dirichlet boundary condition. We notice that
			\[(1-C\alpha^{-\frac12})\|u\|\leq\|\tilde u\|\leq (1+C\alpha^{-\frac12})\|u\|\,.\]
			This tells us that, when $u$ runs over $E_n(\alpha,\xi)$, $\tilde u$ also runs over a space of dimension $n$.
			
			In the same way, we get
			\[\|u'\|^2+\|(\xi+b x_2)u\|^2+b\|u\|^2\geq(1-C\alpha^{-\frac12})\left( \|\tilde u'\|^2+\|(\xi+b x_2)\tilde u\|^2+b\|\tilde u\|^2\right)\,.\]
			We deduce that
			\[\|\tilde u'\|^2+\|(\xi+b x_2)\tilde u\|^2+b\|\tilde u\|^2\leq (1+C\alpha^{-\frac12})\nu^-_n(\alpha,\xi)\|\tilde u\|^2\,.\]
			Using the min-max principle, we infer that
			\[\nu^{\mathrm{Dir}}_n(b,\xi)\leq (1+C\alpha^{-\frac12})\nu^-_n(\alpha,\xi)\,,\]
			and the result follows.
		\end{proof}

		
		\begin{lemma}\label{lem.alphato0}
			For $n\ge 1$,  we have	
			\begin{align}
				&\lim_{\alpha\to 0}\nu^+_n(\alpha,\xi)=\nu_n^{+}(0,\xi)=
				\begin{cases}
					0  & n =1 \\
					\nu_{n-1}^{\mathrm{Dir}}(b,\xi)& n \ge 2
				\end{cases}\,,
				\\
				&\lim_{\alpha\to +\infty}\nu^+_n(\alpha,\xi)=\nu_n^{{\mathrm Dir}}(-b,-\xi)\,.\label{eq.112}
			\end{align}
		\end{lemma}
		\begin{proof}
			The proof is similar to that of Lemma \ref{lem.q-a0}. We have $q^+_{0,\xi}(u)=\|d_\xi^\ad  u\|^2$.
			In particular, $0$ is an eigenvalue associated with $x_2\mapsto e^{-\frac{1}{b}(\xi+bx_2)^2}$. So, $\nu^+_1(0,\xi)=0$. Then, let us consider a positive eigenvalue $\nu$. We have
			\[d_\xi d_\xi^\ad  u=\nu u\,,\quad d^\ad_\xi u(0)=0\,.\]
			This implies that
			\[d^\ad_\xi d_\xi v=\nu v \,,\quad \mbox{ with } v=d^\ad_\xi u\neq 0\,,\mbox{ and } v(0)=0\,.\]
			Conversely, if $v$ is an eigenfunction of  the Dirichlet realization of $d^\ad_\xi d_\xi$ with eigenvalue $\nu$, we have
			\[d_\xi d_\xi^\ad  u=\nu u\,,\quad \mbox{ with } u=d_\xi v\,,\quad \mbox{ and }\quad  d_\xi^\ad  u(0)=0\,.\] The argument to show the limit in \eqref{eq.112} follows the same lines as the proof of \eqref{eq.111}.
		\end{proof}
		\subsection{Study of $\xi\mapsto \nu_n^{\pm}(\alpha,\xi)$}
		\begin{lemma}\label{lem.studynupm}
			Let $\nu^{\pm}(\alpha,\xi)$ be an eigenvalue as in \eqref{eq.eve}. Then, 	we have
			\begin{align}
				\label{eq.der}
				\partial_\xi \nu^{\pm}(\alpha,\xi)=\frac1b\left(\nu^\pm(\alpha,\xi)+\alpha^2\mp2\alpha\xi\right)(u_{\alpha,\xi}^\pm(0))^2\,.
			\end{align}
			In addition, if $\xi_\alpha$ is a critical point of $\xi\mapsto\nu(\alpha,\xi)$, we have
			\begin{align}
				\label{eq.dd}
				\partial_\xi^2\nu^\pm(\alpha,\xi_\alpha)=\mp\alpha\frac2b
				(u_{\alpha,\xi}^\pm(0))^2\,.
			\end{align}
			In particular, $\nu^\pm(\alpha,\cdot)$ has at most one critical point. This critical point can only be a local maximum for $\nu^+(\alpha,\cdot)$ and a local minimum for $\nu^-$. 
		\end{lemma}
		\begin{proof}
			We give again the proof only for the $+$ case. We use the notation from the  proof of the previous Lemma~\ref{lem.qinalpha}. (We also replace $x_2$ by $t$ in  the notation.)	
			
			 Observe that since 
			$\partial_\xi d^\ad u=d^\ad \partial_\xi u +u$, we get
			\begin{align}\label{bxi}
				&(d^\ad\partial_\xi u)(0)
				=\alpha(\partial_\xi u)(0)-u(0)\,.
			\end{align}
			By differentiating  \eqref{eq.eve} with respect to $\xi$ 
			we get
			\begin{align*}
				(d d^\ad -\nu)\partial_\xi u=[\partial_\xi \nu-2(\xi+bt)] u\,.
			\end{align*}
			Hence, \eqref{eq.intbyparts} and \eqref{bxi} yield  
			\begin{align}\label{eq.p}
				\partial_\xi \nu=\langle u,2(\xi+bt)  u \rangle-u(0)^2\,.
			\end{align}
			In addition, integrating by parts and using \eqref{eq.eve}, we calculate
			\begin{align*}
				\langle u,2(\xi+bt)  u \rangle&=\frac{1}{b} \int_{\R_+} u(t)^2\partial_t
				(\xi+bt)^2\,{\rm d}t\\
				&= -\frac{\xi^2}{b}u(0)^2-
				\frac{2}{b}\int_{\R_+}  
				u'(t) 
				(\xi+bt)^2u(t)\,{\rm d}t\\
				&=
				-\frac{\xi^2}{b}u(0)^2-
				\frac{2}{b}\int_{\R_+}  
				u'(t) 
				( \nu+b+\partial_t^2)u(t)\,{\rm d}t\\ 
				&=-\frac{\xi^2}{b}u(0)^2-
				\frac{1}{b}\int_{\R_+}  \partial_t
				[((\nu+b)u(t)^2 +u'(t)^2)]
				\,{\rm d}t\,.
			\end{align*}
			Using that $ u'(0)=(\alpha-\xi)u(0)$ 
			we readily obtain \eqref{eq.der}. 
			Hence, if a critical point $\xi_\alpha$ exists, it satisfies
			$\nu^+(\alpha,\xi_\alpha)+\alpha^2-2\alpha\xi_\alpha=0$. Taking the derivative of \eqref{eq.der} with respect to $\xi$ and evaluating at $\xi_\alpha$ we obtain \eqref{eq.dd}.
		\end{proof}
		With the help of the perturbation theory, we get the following (see   \cite[Lemma 4.14]{barbaroux:hal-02889558}).
		\begin{lemma}\label{lem.limnupm}
			We have
			\[\lim_{\xi\to-\infty}\nu^+_n(\alpha,\xi)=(2n-2)b\,,\quad \lim_{\xi\to+\infty}\nu^+_n(\alpha,\xi)=+\infty\,,\]	
			and
			\[\lim_{\xi\to-\infty}\nu^-_n(\alpha,\xi)=2nb\,,\quad \lim_{\xi\to+\infty}\nu^-_n(\alpha,\xi)=+\infty\,.\]	
		\end{lemma}
		This allows us to show the following.
		\begin{lemma}
			The function $\nu^+(\alpha,\cdot)$ has no critical points. Moreover, $\nu^-(\alpha,\cdot)$ has a unique critical point, which is a global minimum.
		\end{lemma}
		
		\begin{proof}
			Since $\nu^+(\alpha,\cdot)>0$, from \eqref{eq.der} we see that it is increasing on $(-\infty,0)$. If it has a (unique) critical point for some $\xi_\alpha>0$, it must be a non-degenerate global maximum. This contradicts the limit at $\xi=+\infty$, hence $\nu^+(\alpha,\cdot)$ is increasing on $\R$. 
			
			Now assume that  $\nu^-(\alpha,\cdot)$ does not have critical points. From \eqref{eq.der} we must have $\nu^-(\alpha,\xi)+\alpha^2+2\alpha\xi< 0$ for all $\xi$ (since it is the case for $\xi\to-\infty$). But this would imply that $\nu^-(\alpha,\cdot)$ is decreasing on $\R$, which contradicts its limit  at $\xi\to+\infty$.
		\end{proof}
		\section{Proofs for the energy dispersion curves}\label{sec.proofs}
		In this section we start by proving the  characterization described in Theorem~\ref{prop.nupmlambda2}. Next, we apply that result
		to show Theorems~\ref{thm.main} and \ref{prop.limitgamma}. 
		\begin{proof}[Proof of Theorem~\ref{prop.nupmlambda2}]	
			It is enough to deal with the positive eigenvalues of $\mathscr{D}_{\gamma,\xi}(b)$. Indeed, due to the charge conjugation \eqref{eq.symfiber},  $\vartheta_n^-(\gamma,\xi)$ is the $n$-th positive eigenvalue of $\mathscr{D}_{\gamma^{-1},-\xi}(-b)$. Thus, if the characterization \eqref{eq.nu+=lambda2} is established, $\vartheta_n^-(\gamma,\xi)$ is the unique positive solution of  $\nu_n^+(-b,\gamma^{-1}\lambda,-\xi)=\lambda^2$ or equivalently of \eqref{eq.nu-=lambda2} (here we use \eqref{jmb1}).
			
			Let us now prove that \eqref{eq.nu+=lambda2} has exactly one positive solution. Remember that $\gamma>0$. We let
			\[f(\lambda)=	\nu^+_n(\gamma\lambda,\xi)-\lambda^2\,,\]
			and notice that $\lim_{\lambda\to+\infty}f(\lambda)=-\infty$, $f(0)\geq 0$, and $f'(0)>0$ (see \eqref{eq.dalphanu}). Thus \eqref{eq.nu+=lambda2} has at least one positive solution. If $E$ is such a solution, we have $f(E)=0$ and we notice that
			\[f'(E)=\gamma\partial_\alpha \nu_n^+(\gamma E,\xi)-2E=\gamma \left[u_{\gamma E,\xi}(0)\right]^2-2E\,.\]
			To get the sign of $f'(E)$, we consider the polynomial of degree two given by
			\[P(\lambda)=q^+_{\gamma\lambda,\xi}(u_{\gamma E,\xi})-\lambda^2\,.\]
			Because $P(0)\geq 0$, $P(-\infty)=-\infty$, and $P(E)=f(E)=0$ with $E>0$, the polynomial must have two roots of opposite sign.  Thus, $f'(E)=P'(E)<0$.
			This shows that \eqref{eq.nu+=lambda2} has at most one positive solution and thus exactly one, which is denoted by $E_n(\gamma,\xi)$.
			
			In fact, $(E_n)_{n\geq 1}$ is increasing. Indeed,  $$0=\nu_{n+1}^+(\gamma E_{n+1},\xi)-E_{n+1}^2>\nu_{n}^+(\gamma E_{n+1},\xi)-E_{n+1}^2=f(E_{n+1}),$$ which implies that $E_{n}<E_{n+1}$.
			
			 For all $n\geq 1$, due to  Lemma \ref{lem.bij}, $E_n(\gamma,\xi)$ is a positive eigenvalue of $\mathscr{D}_{\gamma,\xi}$. This tells us that
			\[\mathscr{A} : \mathbb{N}^{*}\ni n\mapsto E_n(\gamma,\xi)\in\mathrm{spec}(\mathscr{D}_{\gamma,\xi})\cap\mathbb{R}_+\] 
			is well-defined (and it is injective). 
   
   We now show that the map is surjective. For all $n\geq 1$,  Lemma \ref{lem.bij} implies that $\mathfrak{h}^+_{\gamma\vartheta^+_n(\gamma,\xi),\xi}-(\vartheta^+_n(\gamma,\xi))^2$ has a non-zero kernel. This means that, for some $m\geq 1 $, we have
			\[\nu^+_m(\gamma\vartheta^+_n(\gamma,\xi),\xi)=(\vartheta^+_n(\gamma,\xi))^2\,,\]
			and thus $\vartheta^+_n(\gamma,\xi)=E_m(\gamma,\xi)$. This implies that $\mathscr{A}$ is bijective, hence $E_n(\gamma,\xi)=\vartheta^+_n(\gamma,\xi)$ for all $n\geq 1$. 
		\end{proof}
		
		\subsection{Proof of Theorems \ref{thm.main} and  \ref{prop.limitgamma}}
		In what follows, in order to ease the readability, we drop the reference to the index $n$ in the notation.
		In view of Theorem~\ref{prop.nupmlambda2}
		we have
		\[	\nu^+ (\gamma\vartheta^+,\xi)=(\vartheta^+)^2\]
		and, due to the analyticity and the chain rule, the derivative with respect to $\xi$ gives (with $\alpha\equiv \gamma\vartheta^+)$
		\begin{equation}\label{eq.critheta}
			(\gamma\partial_\alpha \nu^+(\gamma\vartheta^+,\xi)-2\vartheta^+)\partial_\xi\vartheta^++\partial_\xi\nu^+(\gamma\vartheta^+,\xi)= 0\,,
		\end{equation}
		and differentiating with respect to $\gamma$ yields:
		\begin{equation}\label{eq.critheta2}
			(\gamma\partial_\alpha \nu^+(\gamma\vartheta^+,\xi)-2\vartheta^+)\partial_\gamma\vartheta^++\partial_\alpha\nu^+(\gamma\vartheta^+,\xi)\vartheta^+=0\,.
		\end{equation}
		We saw in the proof of Theorem~\ref{prop.nupmlambda2} that $$\gamma\partial_\alpha \nu^+(\gamma\vartheta^+,\xi)-2\vartheta^+<0\,.$$ In the $+$ case, we see that $\vartheta^+(\gamma,\cdot)$ has no critical points and is increasing. We also see that $\gamma\mapsto\vartheta^+(\gamma,\xi)$ is increasing (by using \eqref{eq.dalphanu}), which proves the monotonicity in $\gamma$ of $\vartheta^+$ announced in Theorem \ref{prop.limitgamma}.

		In the $-$ case, by performing the same derivatives on  \eqref{eq.nu-=lambda2}, we have
		\begin{equation}\label{eq.critheta'}
			(\gamma^{-1}\partial_\alpha \nu^-(\gamma^{-1}\vartheta^-,\xi)-2\vartheta^-)\partial_\xi\vartheta^-+\partial_\xi\nu^-(\gamma^{-1}\vartheta^-,\xi)= 0\,,
		\end{equation}
		and
		\begin{equation}\label{eq.critheta2'}
			(\gamma^{-1}\partial_\alpha \nu^-(\gamma^{-1}\vartheta^-,\xi)-2\vartheta^-)\partial_\gamma\vartheta^--\gamma^{-2}\partial_\alpha\nu^-(\gamma^{-1}\vartheta^-,\xi)\vartheta^-=0\,.
		\end{equation}
		We still have $\gamma^{-1}\partial_\alpha \nu^-(\gamma^{-1}\vartheta^-,\xi)-2\vartheta^-<0$. In particular, $\gamma\mapsto\vartheta^-(\gamma,\xi)$ is decreasing.
		
		If $\xi_\gamma$ is a critical point of $\vartheta^-(\gamma,\cdot)$, then we have
		\[\partial_\xi\nu^-(\gamma^{-1}\vartheta^-(\gamma,\xi_\gamma),\xi_\gamma)=0\,.\]
		\begin{remark}
			We recall Lemma \ref{lem.studynupm} and we have
			\[\nu^-(\alpha,\xi_\gamma)+\alpha^2+2\alpha\xi_\gamma=0\,,\quad\mbox{ with } \alpha=\gamma^{-1}\vartheta^-(\gamma,\xi_\gamma)\,.\]
			Hence, $\vartheta^-(\gamma,\xi_\gamma)=-\frac{2\gamma}{\gamma^2+1}\xi_\gamma$ and $\xi_\gamma=\xi_\alpha$ by the uniqueness of the critical point.   
		\end{remark}
		Being a non-degenerate minimum, it is necessary that $\partial^2_\xi\nu^-(\gamma^{-1}\vartheta^-(\gamma,\xi_\gamma),\xi_\gamma)>0$ and, by taking one more derivative in $\xi$ of \eqref{eq.critheta'}, we see that $\partial_\xi^2\vartheta^-(\gamma,\xi_\gamma)>0$. Therefore, all the critical points of $\vartheta^-(\gamma,\cdot)$ are local non-degenerate minima and thus there is at most one such point. If there is no critical point, we have, for all $\xi$,
		\[ \partial_\xi \vartheta^-(\gamma,\xi)=C\, \Big ((1+\gamma^{-2})\vartheta^-(\gamma,\xi)+2\gamma^{-1}\xi\Big )\neq 0\,, \quad C>0.\]
		Let us assume for the moment that \eqref{eq.limitheta-} is true; we will prove that later on. If $\xi$ is sufficiently negative, then the left-hand side of the above expression is negative, so it must remain negative for all $\xi$. This implies that $\vartheta^-(\gamma,\xi)$ must be bounded from above, contradicting  the limit $\xi\to +\infty$ in \eqref{eq.limitheta-}. This ends the analysis of critical points announced in Theorem \ref{thm.main}.

  It remains to explain why \eqref{eq.limitheta-} holds. We only consider the limit $\xi\to-\infty$. We recall Lemma \ref{lem.limnupm}. Let us fix $\varepsilon>0$ and define $\lambda_1=\sqrt{2nb-\varepsilon}$ and $\lambda_2=\sqrt{2nb+\varepsilon}$. Then there exists $\xi(\epsilon)<0$ such that for all $\xi<\xi(\epsilon)$  we have 
  $$\nu^-(\gamma^{-1}\lambda_1,\xi)-\lambda_1^2>0\quad {\rm and}\quad \nu^-(\gamma^{-1}\lambda_2,\xi)-\lambda_2^2<0.$$
  This implies that 
  $$\lambda_1<\vartheta^-(\gamma,\xi)<\lambda_2,\quad \forall \xi<\xi(\epsilon).$$
  The limit $\xi\to+\infty$ can be analyzed similarly (as well as the limits for the $\vartheta^+$). This ends the proof of Theorem \ref{thm.main}.
		
		It remains to discuss the limits in Theorem \ref{prop.limitgamma}. Let us consider  $\vartheta^+_n(\gamma,\xi)$. Take $\varepsilon>0$ and $\lambda=\vartheta_n(0,\xi)-\varepsilon$ (for $n\geq 2$). We have $\nu^+(\gamma\lambda,\xi)-\lambda^2>0$ for $\gamma$ small enough since $\nu^+(0,\xi)=(\vartheta^+(0,\xi))^2$. Thus, $\vartheta^+(0,\xi)-\varepsilon<\vartheta^+(\gamma,\xi)$. In the same way, we get $\vartheta^+(\gamma,\xi)<\vartheta^+(0,\xi)+\varepsilon$. This proves the first limit in Theorem \ref{prop.limitgamma}.

  Next,  we consider the limit $\gamma\to+\infty$. We take $\lambda=\vartheta^+(+\infty,\xi)-\varepsilon$. We have $\nu^+(\gamma\lambda,\xi)-\lambda^2>0$ for $\gamma$ large enough since $\nu^+(+\infty,\xi)=(\vartheta^+(+\infty,\xi))^2$. Thus, $\vartheta^+(+\infty,\xi)-\varepsilon<\vartheta^+(\gamma,\xi)$. We easily get the upper  bound $\vartheta^+(\gamma,\xi)<\vartheta^+(+\infty,\xi)+\varepsilon$. The case of $\vartheta^-(\gamma,\xi)$ is similar.
		\appendix
		\section{Lemma on trace class operators}
		\begin{lemma}\label{tr-cl}Let $\{T_n\}_{n\geq 1}$ be a sequence of trace class operators on some separable Hilbert space, having the property that their trace norms are uniformly bounded, i.e. $\sup_{n\geq 1}\Vert T_n\Vert_{1}= c<\infty$. Assume that $T_n$ converges to $T$ in the operator norm topology. Then $T$ is trace class.  
		\end{lemma}
		\begin{proof}
			Since the $T_n$'s are compact operators, $T$ is also compact and admits a singular value decomposition (SVD), i.e. there exist two orthonormal systems $\{f_j\}_{j\geq 1}$ and   $\{g_j\}_{j\geq 1}$, together with a set of non-increasing singular values $s_j\geq 0$ such that 
			$$T=\sum_{j\geq 1} \, s_j\, |f_j\rangle \,\langle g_j|.$$
			
			$T$ is trace class if $\sum_{j\geq 1}s_j<\infty$. We will show that for every $N\geq 1$ we have $\sum_{j=1}^Ns_j\leq c$. Let us introduce the SVD of each $T_n$ as 
			$$T_n=\sum_{j\geq 1} \, s_j^{(n)}\, |f_j^{(n)}\rangle \,\langle g_j^{(n)}|,\quad \Vert T_n\Vert_1=\sum_{j\geq 1} s_j^{(n)}\, \leq c.$$
			Then 
			\begin{align*}
				\sum_{k=1}^Ns_k=\lim_{n\to\infty}\sum_{k=1}^N  \langle  f_k,\, T_n\, g_k\rangle =\lim_{n\to\infty} \sum_{k=1}^N\sum_{j\geq 1} \, s_j^{(n)} \langle f_k,f_j^{(n)}\rangle \,\langle g_j^{(n)}, g_k\rangle .
			\end{align*}
			Using Bessel's inequality for the orthonormal systems $f_k$ and $g_k$ we have 
			$$\Big | \sum_{k=1}^N \langle f_k,f_j^{(n)}\rangle \,\langle g_j^{(n)}, g_k\rangle\Big |\leq \sqrt{\sum_k |\langle f_k,f_j^{(n)}\rangle |^2} \, \sqrt{\sum_k |\langle g_k,g_j^{(n)}\rangle |^2}\leq \Vert f_j^{(n)}\Vert \, \Vert g_j^{(n)}\Vert =1$$
			which holds for all $j,n,N\geq 1$. 
			This implies 
			\begin{align*}
				\Big |\sum_{k=1}^N\sum_{j\geq 1} \, s_j^{(n)} \langle f_k,f_j^{(n)}\rangle \,\langle g_j^{(n)}, g_k\rangle \Big |\leq \sum_{j\geq 1} \, s_j^{(n)}=\Vert T_n\Vert_1\leq c,
			\end{align*}
			hence $\sum_{j=1}^\infty s_j<\infty $ and $T$ is trace class.
			
		\end{proof}

	\subsection*{Acknowledgments}
	This work was partially conducted within the France 2030 framework programme, the Centre Henri Lebesgue  ANR-11-LABX-0020-01. The authors thank T. Ourmi{\`e}res-Bonafos for useful discussions. They are also very grateful to the
CIRM (and its staff) where this work was initiated.  E.S.  acknowledge  support from  Fondecyt (ANID, Chile) 
through the  grant \# 123--1539.
This work has been partially supported by CNRS International Research Project Spectral Analysis of
Dirac Operators – SPEDO. H.C. acknowledge support from the Independent Research
Fund Denmark–Natural Sciences, grant DFF–10.46540/2032-00005B.

	\bibliographystyle{abbrv}      
	\bibliography{cinq}

\end{document}